\documentclass[journal]{IEEEtran}

\usepackage{amsmath}
\usepackage{amsthm}
\usepackage{amsfonts}
\usepackage{amssymb}
\usepackage{epsfig}
\usepackage{psfrag}
\usepackage{graphicx}
\usepackage{float}
\usepackage{multirow}
\usepackage{tabularx}
\usepackage{color}
\usepackage{algorithm}
\usepackage{algorithmic}
\usepackage{lineno}
\usepackage{hyperref}
\usepackage{multicol}
\usepackage{commath}
\usepackage{stfloats}
\usepackage{wrapfig}
\usepackage{subcaption}
\usepackage{cleveref}

\newtheorem{lemma}{Lemma}

\newcolumntype{C}[1]{>{\centering\let\newline\\\arraybackslash\hspace{-2mm}}m{#1}}

\begin{document}

\date{}

\title{Towards Secure Localization in Randomly Deployed Wireless Networks}

\author{Marko~Beko
		and Slavisa~Tomic
\thanks{The authors are with COPELABS, Universidade Lus\'{o}fona de Humanidades e Tecnologias, Campo Grande 376,  1749 - 024 Lisboa, Portugal. (e-mails: slavisa.tomic@ulusofona.pt, beko.marko@ulusofona.pt).}}

\maketitle



\begin{abstract}
Being able to accurately locate wireless devices, while guaranteeing high-level of security against spoofing attacks, benefits all participants in the localization chain (e.g., end users, network operators, and location service providers). On the one hand, most of existing localization systems are designed for innocuous environments, where no malicious adversaries are present. This makes them highly susceptible to security threats coming from interferers, attacks or even unintentional errors (malfunctions) and thus, practically futile in hostile settings. On the other hand, existing secure localization solutions make certain (favorable) assumptions regarding the network topology (e.g., that the target device lies within a convex hull formed by reference points), which restrict their applicability. Therefore, this work addresses the problem of target localization in randomly deployed wireless networks in the presence of malicious attackers, whose goal is to manipulate (spoof) the estimation process and disable accurate localization. We propose a low-complex solution based on clustering and weighted central mass to detect attackers, using only the bare minimum of reference points, after which we solve the localization problem by a bisection procedure. The proposed method is studied from both localization accuracy and success in attacker detection point of views, where closed-form expressions for upper and lower bounds on the probability of attacker detection are derived. Its performance is validated through computer simulations, which corroborate the effectiveness of the proposed scheme, outperforming the state-of-the-art methods.
\end{abstract}




\section{Introduction}
\label{sec:intro}

So far, wireless localization techniques have been developed independently from the communication ones. This is likely to change in the forthcoming fifth generation (5G) networks, where dense access point deployment and large bandwidths are expected to allow for accurate localization, together with low energy consumption~\cite{Lohan:2018,Bourdoux:2020, Witrisal:2019}. On the one hand, highly-accurate localization can bring benefits to both network operators and end users, with techniques like location-aware interference mitigation, power and latency optimized end-to-end communications, and user-personalized location based services~\cite{Lohan:2018} to name a few. On the other hand, it might raise serious privacy and security issues, since various malicious attacks~\cite{Simeone:2016} or inaccurate data acquirement due to equipment malfunctions could be possible to occur. Therefore, an additional requirement for future localization systems is that the estimation process is carried out securely~\cite{Ranganathan:2017, Avoine:2018}.

Conventional localization systems~\cite{Gao:2016}-\cite{Tomic_TVT:2020} are designed for benign environments where no security requirements are needed; thus, they are susceptible to spoofing attacks and employing them in applications where adversaries (able to manipulate (spoof) locations/measurements of other nodes) are present could lead to disastrous consequences (e.g., failure in rescue tasks, drone/self-driving cars collisions, etc.).

A secure-ranging problem has been addressed in the literature, where the focus is on developing strategies to verify that a node is at a claimed distance from another one. Verifiable multilateration, location verification using mobile base stations, and other distance bounding protocols were proposed to withstand attacks \cite{Capkun:2006}, \cite{Singh:2019}. In \cite{Capkun:2006}, the problem of distance reduction attacks was addressed, where distance bounding protocol was employed, together with simple challenge-response schemes. The work in \cite{Singh:2019} addressed the problem of distance enlargement attacks, where a novel modulation scheme based on distance commitment-verification protocol was proposed, in order to detect the attacker. Even though the methods in \cite{Capkun:2006}, \cite{Singh:2019} guarantee high security levels under the assumption that targets are within a convex hull formed by reference points and that direct distance measurements are available, there are countless applications in which at least one of these assumptions does not hold, where \cite{Capkun:2006}, \cite{Singh:2019} could not guarantee security.

Secure localization problem in wireless networks in the presence of malicious adversaries has also attracted attention in the research community \cite{DLiu:2008}-\cite{Won:2019}. A greedy approach to find the location consistent with the largest number of measurements from reference points was explored in \cite{DLiu:2008}. A voting-based scheme was introduced, in which the localization area is divided into a grid and the vote count of each grid point is incremented if its distance from a reference point is similar to the distance measurement from the reference point. An attack-resistant and device-independent method was developed in \cite{He:2009}, where an attack-driven model was specified by using Petri net. Both distance reduction attacks and distance enlargement attacks were considered. In \cite{Garg:2012}, an iterative gradient descent technique with inconsistent measurement pruning was proposed. The work considered mobile sensor networks where some nodes transmit false information. To account for the possibility of malicious nodes, the cost function was updated iteratively by eliminating the reference points with large residues from the localization process. A weighted least squares (WLS) model for localization based on received signal strength (RSS) measurements was proposed in \cite{Mukhopadhyay:2018}, where non-cryptographic uncoordinated attacks were considered. WLS scheme was based on proper weight assignment founded on log-distance model, where nodes \emph{closer} to the target received large weights and vice versa. In \cite{XLiu:2019}, two algorithms that utilize density-based spatial clustering to detect abnormal clusters were proposed, which were further examined via a sequential probability ratio test. First, an adaptive clustering algorithm was performed in order to reduce the number of initial parameters, and avoid situations where local outliers are categorized into normal clusters. Then, a sequential probability ratio test based on consistency characteristics of time of arrival (TOA) and RSS measurements was employed to provide accurate detection results. The work in \cite{Won:2019} introduced two attack models based on the knowledge of target location, namely aligned node location and inside-attack. The former one exploits nodes that are aligned on a line, while the latter one disables degree of consistency filtering algorithm by placing malicious location references inside benign ones. To protect against these attacks the authors in \cite{Won:2019} proposed a novel beacon placement strategy and a filtering technique that can filter out malicious location references introduced by inside-attacks.

A similar problem to detecting attackers in secure localization setting is the problem of link identification in mixed line-of-sight (LOS) and non-LOS (NLOS) environments, since bias introduced in NLOS links can be interpreted to some extent as an attack carried out by a malicious node. The key difference between the two problems is, however, the fact that the attackers can vary their attack intensity or even cooperate between themselves in order to compromise the localization task and avoid being detected, whereas obstacles do not have such capabilities. Nevertheless, the problem of NLOS identification has gained attention in the research society recently, where methods based on hypothesis testing \cite{Zhang:2014}, quasi-Newton method \cite{Yin:2014}, and non-parametric machine learning \cite{Nguyen:2015} were proposed. Recently, a new approach based on variational Bayesian localization (VBL) was introduced in \cite{Li:2020}, where the authors used approximations to the true posterior distributions and applied the variational framework to find the optimal variational distributions. This was done in an iterative fashion by employing a set of particles which might increase the computational burden of a such an approach, but the authors in \cite{Li:2020} also considered imperfect knowledge about the reference points' locations.

In contrast to \cite{Capkun:2006}-\cite{Li:2020}, the present work addresses the problem of target localization in randomly deployed wireless networks in the presence of malicious adversaries, using only the bare minimum of reference points. The proposed algorithm exploits clustering and weighted central mass (WCM) in order to determine an initial target location. It then takes advantage of this solution to determine distance estimates from it to all reference points which, together with threshold-based keying, are used to detect attackers. Hence, the target localization and the attacker detection problems are coupled and addressed in conjunction, rather than independently. The detected attackers are then removed from the estimation process, and the localization problem is converted into a generalized trust region sub-problem (GTRS) framework, which is solved \emph{exactly} by a bisection method. The proposed method is studied from both localization accuracy and success in attacker detection perspectives. It will be shown here that, malicious attacks of relatively low intensity are preferred to be treated as measurement noise, rather than detecting them and excluding the respective reference point from the localization process. This is because low-intensity attacks do not have severe impact on the estimation accuracy and, when the number of reference points is scarce, any mildly-corrupted measurement is a valuable asset that should not be overlooked. As the attack intensity increases, naturally, one desires to accomplish high detection rates in order to exclude corrupted information and enhance the localization accuracy. Therefore, tuning the threshold to balance out the performance for both metrics turns out to play an important, but, as we will see, not a crucial role for the overall performance of the proposed method. Also, we will see that the new method outperforms by far the state-of-the-art approaches, both in terms of localization accuracy and success in attacker detection.

The main contributions of the present work are threefold, and are summarized in the following.
\begin{itemize}

\item{It introduces a novel scheme for detecting corrupted reference points based on a simple geometrical approach and threshold-based keying, which does not depend on the network topology nor prior knowledge about additional parameters (e.g., noise variance).}

\item{It derives theorethical upper and lower bounds on the achievable probability of detection of the proposed detection scheme.}

\item{It presents an \emph{exact} solution to the localization problem based on bisection procedure, which, combined with the proposed detection scheme, results in an efficient and secure localization algorithm that requires only the bear minimum of reference points.}

\end{itemize}

The remainder of this work is organized as follows. In Section~\ref{sec:prob_form}, the measurement and attacker models are introduced, and the localization problem is formalized. Section~\ref{sec:wcm_gtrs} describes the proposed scheme for attacker detection, together with the proposed estimator for target localization. The performance of the proposed method is assessed in Section~\ref{sec:results}, while Section~\ref{sec:conclusions} summarizes the main findings of the work.



\section{Problem Formulation}
\label{sec:prob_form}

Consider a $q$-dimensional ($q = 2$ or $3$) wireless network composed of $N$ anchor (reference) nodes, whose true locations are denoted by $\boldsymbol{a}_i$, $i = 1, ..., N$, and a target, whose true location is denoted by $\boldsymbol{x}$. We assume that the $k$-th distance measurement ($k = 1, ..., K$) between the target node and the $i$-th anchor node is obtained from ultra-wide band (UWB) ranging systems, namely through the two-way time of arrival (TW-TOA) observations, where nodes measure the time of propagation of the radio signal between them, which can be modeled \cite{Gao:2016,Tomic_SPL:2018,Gholami:2016} as
\begin{equation}
t_{i,k} = \frac{\|\boldsymbol{x}-\boldsymbol{a}_i\|}{c} + \frac{T_{i,k}}{2} + \tilde{n}_{i,k},
\label{eq:model_original}
\end{equation}
where $c$ is the propagation speed of the signal, $T_{i,k}$ is the (known) processing time of the signal at the $i$-th anchor (also known as the turn-around time) and $\tilde{n}_{i,k}$ is a random (positive) delay introduced by the target during packet interception, modeled as a positive-mean Gaussian random variable, i.e., $\tilde{n}_{i,k} \sim \mathcal{N}(\tilde{\mu}_{i,k}, \tilde{\sigma}_{i,k}^2)$. For the sake of simplicity and better clarity of the following derivations, let us consider the case where $K=1$ so that we can remove the subscript $k$ in \eqref{eq:model_original}. The TW-TOA ranging based on IEEE 802.15.4a protocol~\cite{DAmico:2010} is illustrated in Fig.~\ref{fig:tw_toa_protocol}.
\begin{figure}
\centering
\hspace*{-0mm}\includegraphics[width=\linewidth]{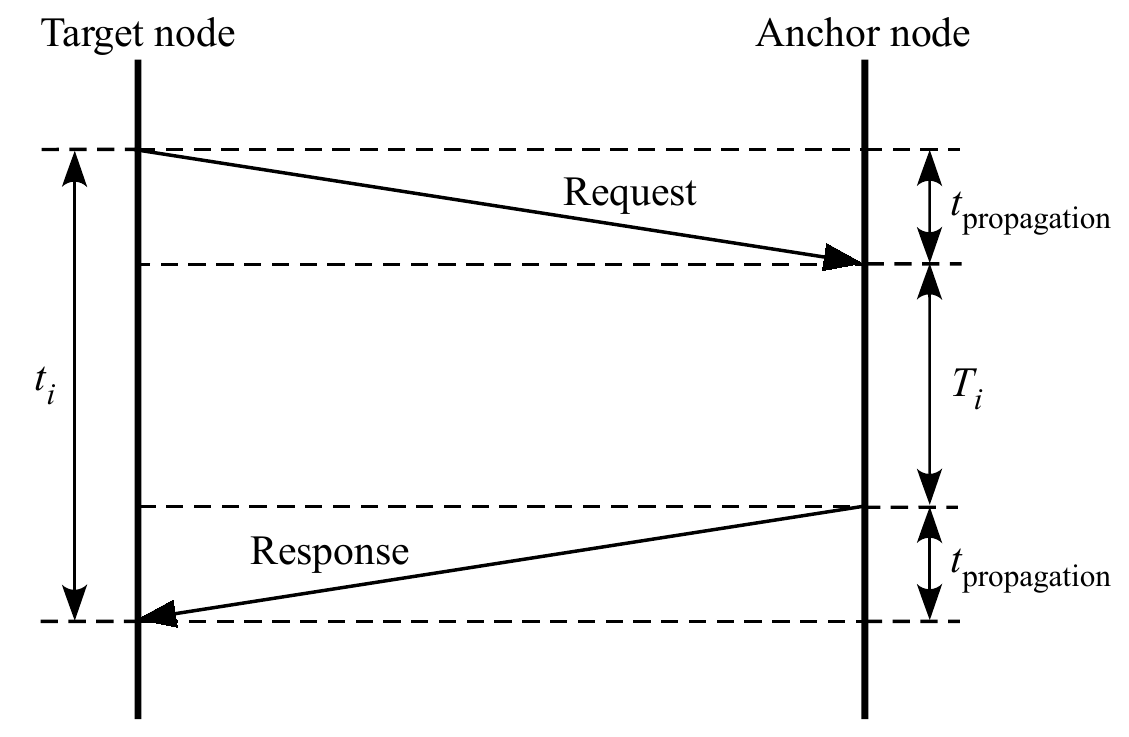}
\vspace*{-5mm}\caption{Illustration of TW-TOA ranging based on IEEE 802.15.4a protocol.}
\label{fig:tw_toa_protocol}
\end{figure}

By multiplying \eqref{eq:model_original} by $c$ and subtracting $c\tilde{\mu}_i$ from both sides, the following range measurement model is obtained
\begin{equation}
d_i = \|\boldsymbol{x}-\boldsymbol{a}_i\| + \frac{d_{T_i}}{2} + n_i,
\label{eq:tw_toa_meas}
\end{equation}
where $d_i = c t_i - c\tilde{\mu}_i$, $d_{T_i} = c T_i$, and $n_i = c(\tilde{n}_i - \tilde{\mu}_i)$, i.e., $n_i \sim \mathcal{N}(0, \sigma_i^2)$, with $\sigma_i = c \tilde{\sigma}_i$. For simplicity, we set $d_{T_i} = 0$, since we assume here that $T_i$ is known. Furthermore, without loss of generality, we assume that $\sigma_1 = \sigma_2 = \hdots = \sigma_N = \sigma$.

In this work, it is assumed that only external attackers are present in the network and that they can spoof a measurement of an anchor node according to
\begin{equation}
\ddot{d}_i = d_i + \delta,
\label{eq:corrup_meas}
\end{equation}
with $\delta \geq 0$ denoting the intensity of the measurement corruption\footnote{Technically, $\delta$ can take on any non-negative value, but it is intuitively clear that extremely large values (e.g., $\delta \gg 1$) would most likely expose the attacker, as we will see in Section~\ref{sec:results}.}. Note that, from the definition in \eqref{eq:corrup_meas}, it follows that the distance measurement of a corrupted anchor node can only be enlarged, since it is physically impossible for the external attacker to reduce it due to the propagation speed of the electromagnetic waves \cite{Capkun:2006}. This type of an attack is even more hazardous than distance-reduction attack, since an adversary only needs to obliterate/distort the authentic signal and replay its delayed version, without compromising any upper-layer protocols during the process~\cite{Singh:2019}.

Given the distance measurements in~\eqref{eq:tw_toa_meas}, one could determine the location of the target node according to the maximum likelihood (ML) criterion~\cite[Ch. 7]{Kay:1993}, as follows
\begin{equation}
\widehat{\boldsymbol{x}} = \underset{\boldsymbol{x}}{\text{arg\,min}} \displaystyle\sum_{i=1}^N \left( d_i - \|\boldsymbol{x}-\boldsymbol{a}_i\| \right)^2.
\label{eq:mle}
\end{equation}

Nevertheless, the problem in~\eqref{eq:mle} is non-convex and does not have a solution in closed-form. Besides, if a measurement is corrupted by an external attacker and one employs it in the localization process, the obtained location estimate could be \emph{far away} from the true one. Therefore, in the following section, we propose a method to first determine which of the distance measurements are being spoofed, followed by an efficient approach to circumvent the non-convexity of~\eqref{eq:mle} without any knowledge about $\sigma$ and estimate the location of the target node by just a bisection procedure.



\section{The Proposed Method}
\label{sec:wcm_gtrs}

This section describes the proposed method for secure localization in random wireless networks. It is organized into three parts: in the first part, the proposed method for attacker detection based on WCM is described, followed by a discussion on the probability of attacker detection presented, while the last part proposes a method for estimating the target's location based on converting the localization problem into a GTRS framework, which is then solved \emph{exactly} by merely a bisection procedure.

\subsection{Attacker Detection}
\label{subsec:attack_det}

For the sake of simplicity, let us start by considering that there is only one attacker in the network; later on, the proposed solution will be extended to a more general case. Because one does not know which of the anchor nodes is corrupted, all distance measurements are treated as only noise-corrupted in the beginning. According to distance measurements and the known locations of anchor nodes, one can form circles, $c_i$, $i = 1, ..., N$, centered at anchor nodes with radii equal to their respective distance measurements to the target and calculate the intersection points of the circles (provided that they exist) as
\begin{equation}
\boldsymbol{p}_{ij} = \boldsymbol{p}_0 \pm \boldsymbol{t}, \,\,\, \text{for} \,\,\, i = 1, ..., N-1, j = i+1, ..., N,
\label{eq:intersections}
\end{equation}
where $\boldsymbol{p}_0 = \frac{\boldsymbol{a}_i + \boldsymbol{a}_j}{2} + \frac{d_i^2-d_j^2}{2 \|\boldsymbol{a}_j - \boldsymbol{a}_i\|^2} \left(\boldsymbol{a}_j - \boldsymbol{a}_i\right)$ and $\boldsymbol{t} = \frac{\sqrt{k}}{2 \|\boldsymbol{a}_j - \boldsymbol{a}_i\|^2} \boldsymbol{M} (\boldsymbol{a}_j - \boldsymbol{a}_i)$, with $k = \left( (d_i + d_j)^2 - \|\boldsymbol{a}_j - \boldsymbol{a}_i\|^2 \right) \left( \|\boldsymbol{a}_j - \boldsymbol{a}_i\|^2 - (d_j - d_i)^2 \right)$ and $\boldsymbol{M} = [0 \,\, -1; 1 \,\,\, 0]$; see Fig.~\ref{fig:circle_int_points}. The intersection points obtained from \eqref{eq:intersections} are then used to form clusters based on their physical proximity. The size of the clusters depends on the number of potential attackers, as explained in the following.
\begin{figure}
\centering
\hspace*{-0mm}\includegraphics[width=\linewidth]{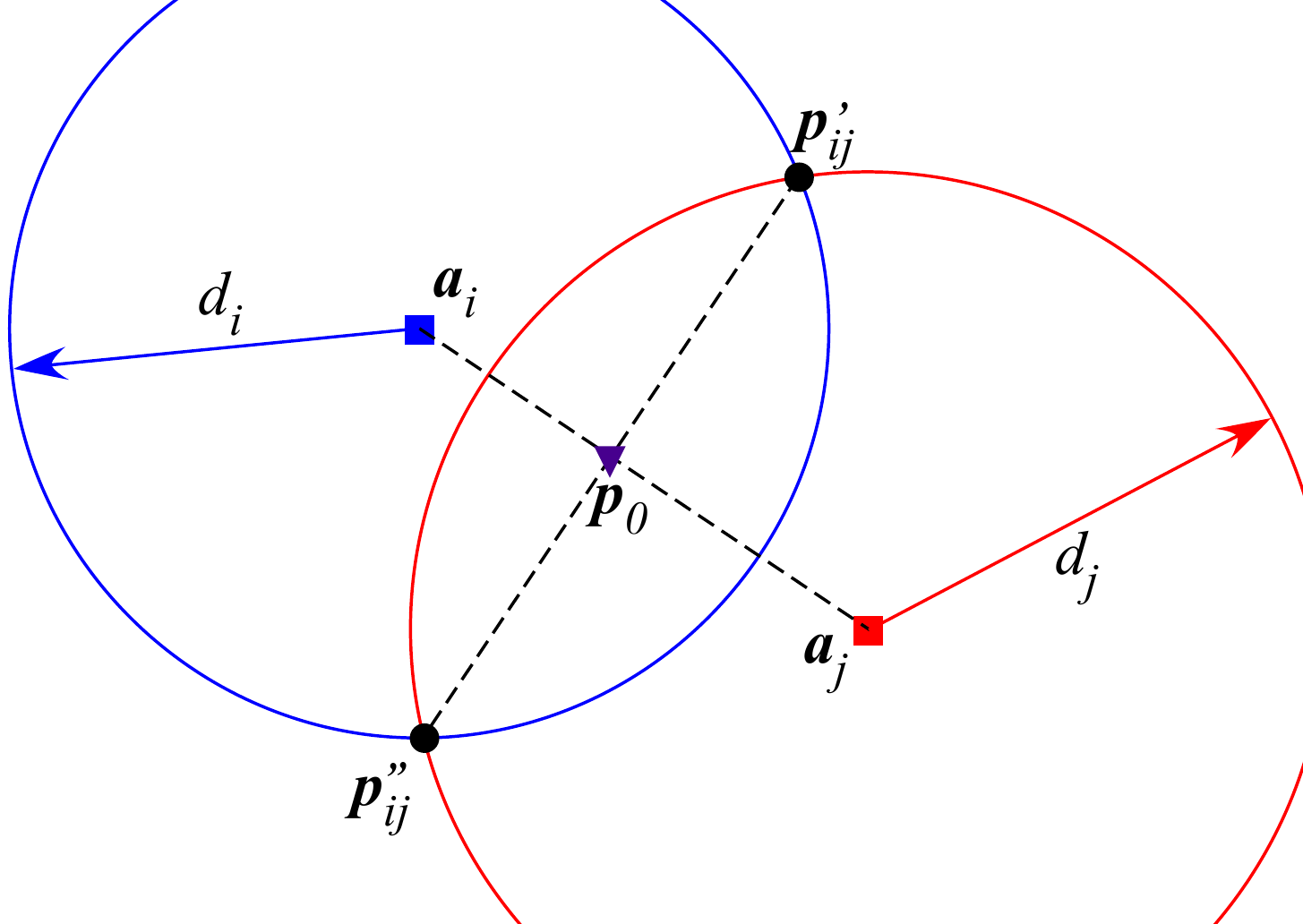}
\vspace*{-5mm}\caption{Illustration of finding the intersection points between a pair of circles.}
\label{fig:circle_int_points}
\end{figure}

For the sake of notation simplicity, we define the set of all anchor nodes as $\mathcal{A} = \{i: 1 \leq i \leq N\}$, the set of all intersection points $\mathcal{P} = \{ \boldsymbol{p}_{ij}: c_i \cap c_j \neq \varnothing \}$, and the tuple set of potentially corrupted anchor nodes, i.e., of anchors whose circles do not intersect as $\mathcal{C} = \{ (i, j): i, j \in \mathcal{A} \wedge c_i \cap c_j = \varnothing \}$, where the notation $c_i \cap c_j = \varnothing$ is used to denote that the circles corresponding to the $i$-th and $j$-th anchor nodes do not intersect.

Three cases are then distinguished: (a) all circles intersect with each other, (b) a circle has no intersection points with any of the remaining ones, and (c) a circle intersects with at least one of the other circle, but not with all of them. These three cases are illustrated in Fig.~\ref{fig:intersections}.
\begin{figure}
\begin{subfigure}{.5\textwidth}
\hspace*{0mm}\includegraphics[width=\textwidth]{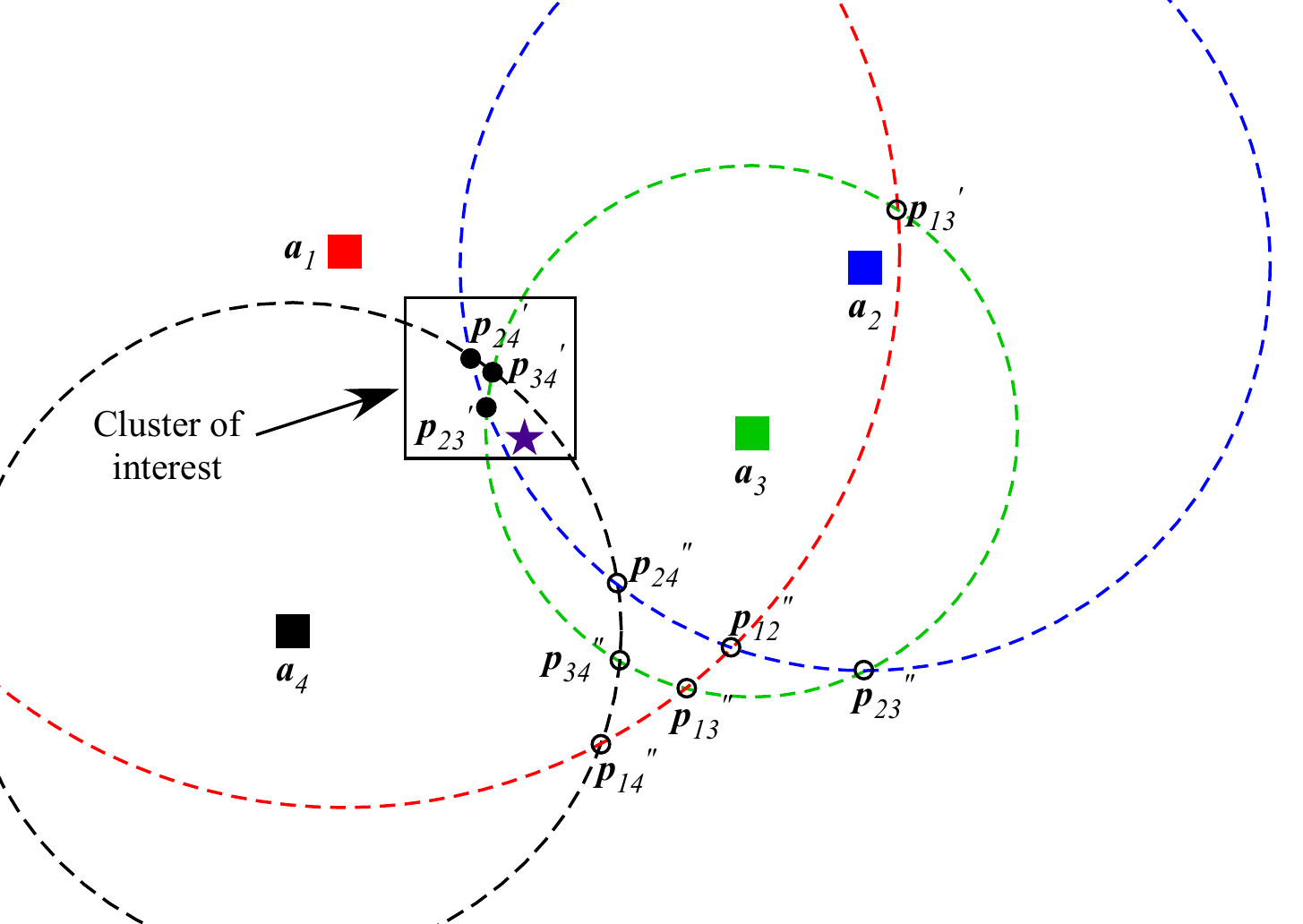}
\caption{All circles intersect each other}
\label{fig:intersections_a}
\end{subfigure}
\vspace*{3mm}
\begin{subfigure}{.5\textwidth}
\hspace*{0mm}\includegraphics[width=\textwidth]{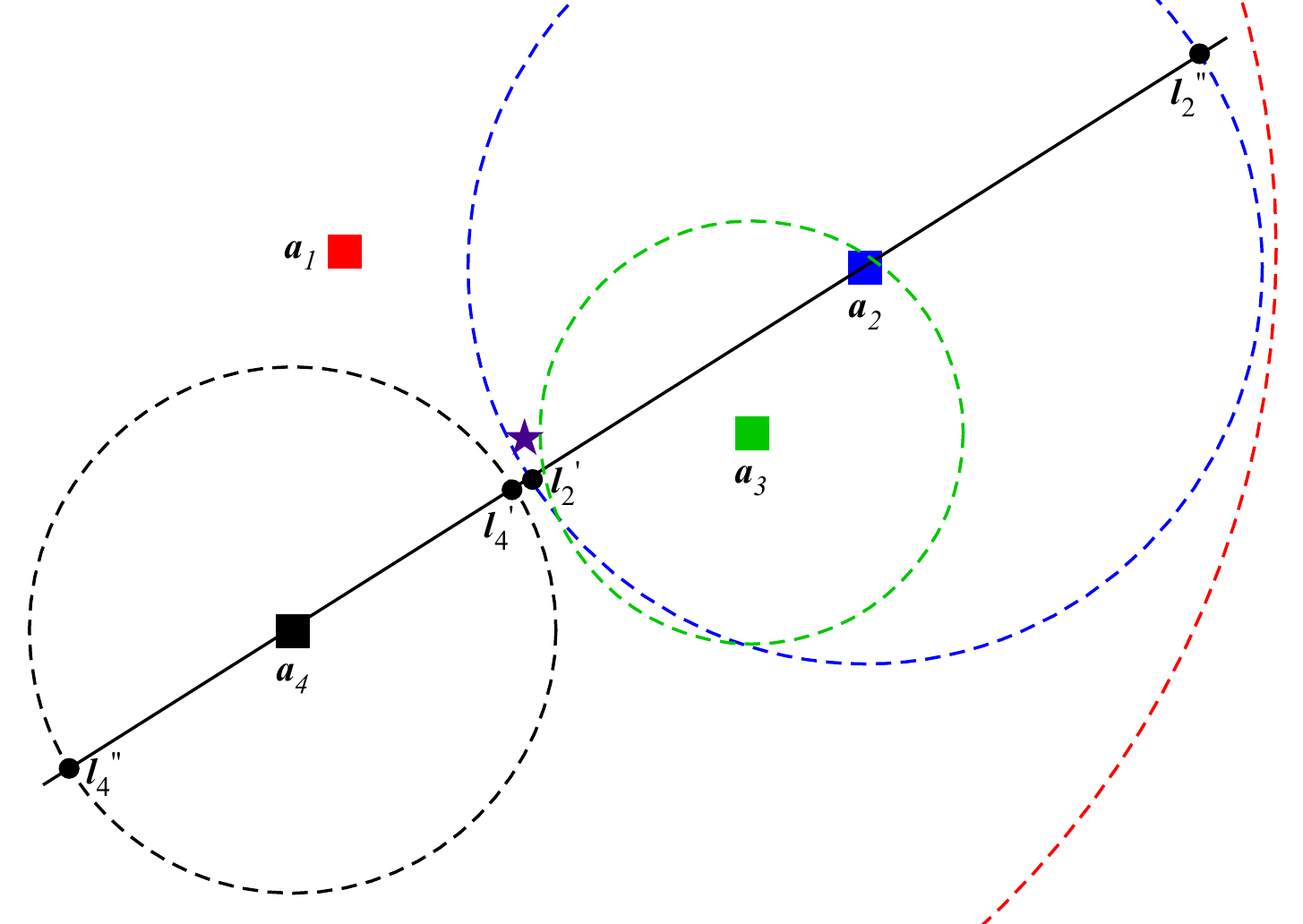}
\caption{A circle has no intersection with other ones}
\label{fig:intersections_b}
\end{subfigure}
\vspace*{3mm}
\begin{subfigure}{.5\textwidth}
\hspace*{0mm}\includegraphics[width=\textwidth]{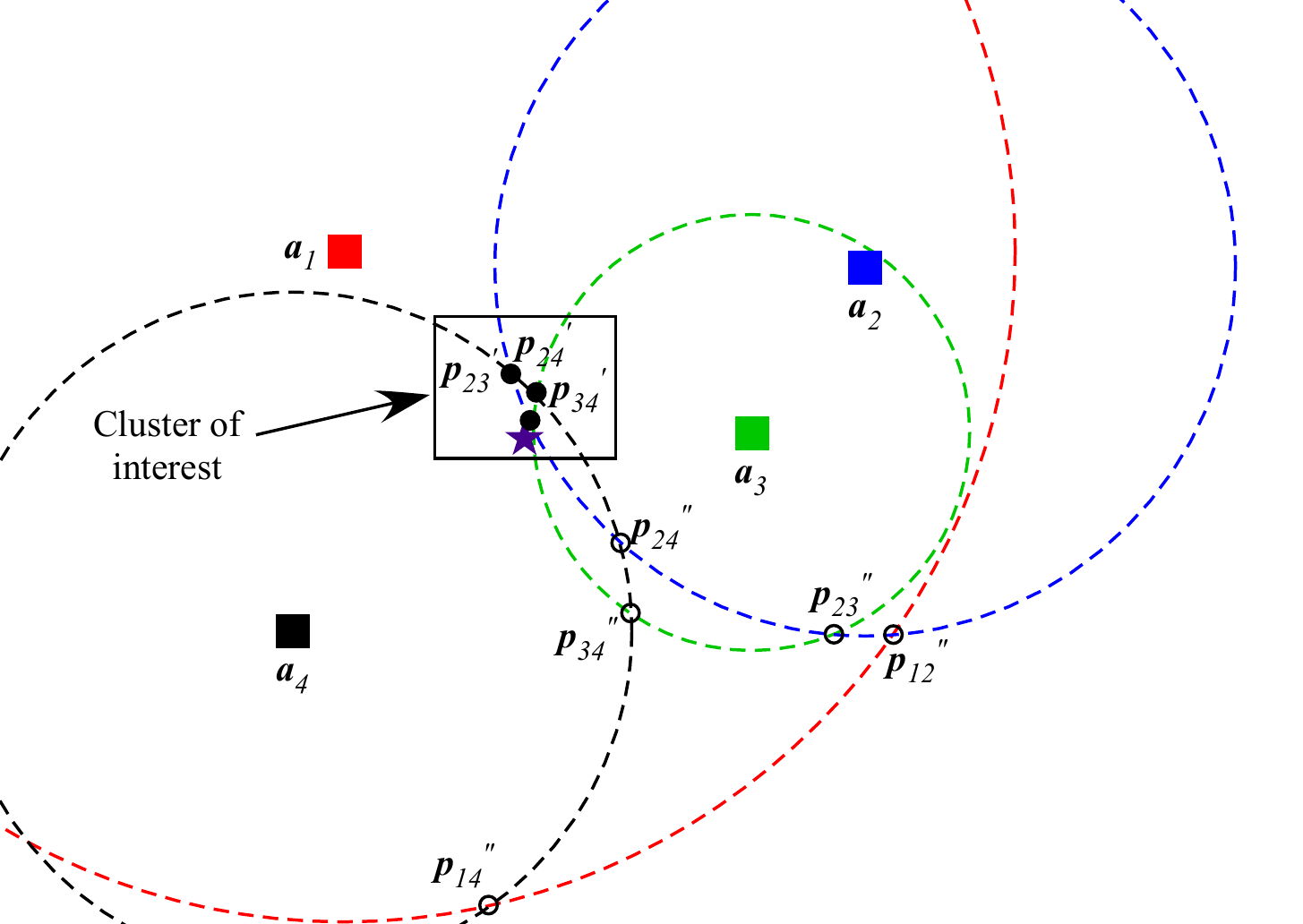}
\caption{A circle intersects some, but not all other ones}
\label{fig:intersections_c}
\end{subfigure}
\vspace*{0mm}
\caption{Illustration of the three considered cases regarding circle intersections.}
\label{fig:intersections}
\end{figure}

Case (a): In the case where all circles intersect, Fig.~\ref{fig:intersections_a}, we simply choose the smallest cluster of size $N-1$, i.e., the closest $N-1$ points to each other. One way to do this is to calculate the Lebesgue measure of a set of points and choose the points with the smallest measure. These points are considered as \emph{honest} points and are stacked in the set $\mathcal{H}=\left\{ \boldsymbol{p}_{ij}: (i,j) \notin \mathcal{C} \wedge \mathcal{P}' \subseteq \mathcal{P} \wedge |\mathcal{P}'| = N-1 \right\}$, with $|\bullet|$ being the cardinality (the number of elements) in a set and $\mathcal{P}'$ denoting the set of $N-1$ \emph{honest} points with minimal Lebesgue measure (combination of $N-1$ points forming the smallest area).

Case (b): In the case where a circle has no intersections with other circles, one only needs to check if an intersection could occur if the radius of the circle is increased. If that is the case, the corresponding anchor node is not considered corrupted, since an attacker can only enlarge its distance measurement; otherwise, the corresponding anchor node is considered corrupted, and the existing intersection points are treated as \emph{honest}. This can easily be checked by drawing a line through the corresponding anchor node and the remaining ones to find the intersection points of the line and a pair of circles, $\boldsymbol{l}_i'$ and $\boldsymbol{l}_i''$ in Fig.~\ref{fig:intersections_b}. Afterwards, it suffices to compare the distance from the center of the circle with no intersections to the two closest intersection points on a pair circles.

Case (c): In the last case, Fig.~\ref{fig:intersections_c}, a set of $N-|\mathcal{C}| \geq q+1$ \emph{honest} points is considered, i.e., $\mathcal{H}=\left\{ \boldsymbol{p}_{ij}: (i,j) \notin \mathcal{C} \wedge \mathcal{P}' \subseteq \mathcal{P} \wedge |\mathcal{P}'| = N-|\mathcal{C}| \right\}$. Naturally, the condition $N-|\mathcal{C}| \geq q+1$ corresponds to the minimum number of points required to localize the target in a $q$-dimensional space.

A preliminary estimate of the target's location can then be obtained as the weighted central mass of \emph{honest} points as
\begin{equation}
\hat{\boldsymbol{x}}^{(1)} = \displaystyle\sum_{(i,j):\boldsymbol{p}_{ij} \in \mathcal{H}} \omega_{ij} \boldsymbol{p}_{ij},
\label{eq:x_est1}
\end{equation}
where $\omega_{ij} = \frac{1/\bar{d}_{ij}}{1/\sum_{(i,j):\boldsymbol{p}_{ij} \in \mathcal{H}}\bar{d}_{ij}}$ with $\bar{d}_{ij} = \frac{d_i+d_j}{2}$, is the weight assigned to each of the \emph{honest} points. The weights are chosen deliberately in this form in order to assign smaller weights to a pair of largest (on average) distance measurements so that the influence of a corrupted anchor node is minimized, in the case that such a node is not exposed so far.

Once the initial target location estimate is available, it is exploited to detect the corrupted anchor node and further enhance the localization accuracy as follows. First, the distance estimates are calculated between $\boldsymbol{a}_i$, $i = 1, ..., N$, and $\hat{\boldsymbol{x}}^{(1)}$ as
\begin{equation}
\hat{d}_i = \|\hat{\boldsymbol{x}}^{(1)} - \boldsymbol{a}_i\|.
\label{eq:dist_est}
\end{equation}
Then, the relative error between the measured distances and the estimated ones is calculated as
\begin{equation}
e_i = \frac{|d_i - \hat{d}_i|}{m_d},
\label{eq:rel_error}
\end{equation}
where $m_d$ denotes the median of the distance measurements $d_i, i = 1, ..., N$.

If $e_m > \tau$, where $e_m = \max\left\{ e_1, ..., e_N \right\}$ and $\tau$ is a pre-defined constant from the interval $[0, 1]$ used to set a threshold in order to \emph{distinguish} between measurement noise and attack, the $m$-th anchor node is marked as an attacker and included in the (singular) set of attackers, i.e., $\mathcal{S} = \{m: m \in \mathcal{A} \, \wedge \, e_m > \tau\}$; otherwise, no attacker is detected. Note that in~\eqref{eq:rel_error}, we opted to divide by the median of all measured distances. The main reason for this choice is that the median is a metric robust to outliers as it is well known, which allows us to avoid dividing by an excessively corrupted distance measurement in~\eqref{eq:rel_error}. Preventing this from happening is of interest since the attacker can only enlarge the measured distance, and employing corrupted distance in~\eqref{eq:rel_error} would likely decrease the value of the relative error, and consequently reduce our chances of detecting the attack. Obviously, opting for the median of all measured distances is not the only choice that would allow us to have better chances of detection, as we could also select the average of the measured distances or even the estimated ones. However, both of these parameters are directly dependent on the corrupted measurement, and thus, extremely corrupted measurements would likely lead to undesired effects.

Note that we could have also opted for employing some sort of likelihood ratio test, such as a generalized likelihood ratio test (GLRT)~\cite{Besson:2017, Xie:2017}, \cite[Ch. 4]{Kay:1998}, to detect the attacker. However, such schemes typically require a tuning of thresholds for a chosen false alarm (probability of detection) in order to obtain the desired probability of detection (false alarm), and might require additional knowledge of model parameters (e.g., noise variance)~\cite{Xie:2017}. Hence, the proposed detection scheme seems intuitive and simple, and does not require any additional knowledge about the model parameters.

\subsection{Probability of Detection}
\label{subsec:prob_det}

Combining the measurement model in~\eqref{eq:tw_toa_meas} with~\eqref{eq:corrup_meas} one gets
\begin{equation}
d_i =
\begin{cases}
      \|\boldsymbol{x}-\boldsymbol{a}_i\| + n_i, & \text{if} \,\, i \neq a\\
      \|\boldsymbol{x}-\boldsymbol{a}_i\| + \delta + n_i, & \text{if} \,\, i = a,\\
\end{cases}
\label{eq:our_model}
\end{equation}
with $a$ denoting the attacker link. The problem of attacker detection is then very similar to the problem of non line-of-sight identification~\cite[Ch. 16.4]{Sharp:2019}. Nevertheless, the proposed approach combines localization and detection problems (since in our procedure the probability of detection depends directly on the solution of WCM), rather than tackling them independently.

Let us define $e_i = \abs{y_i}$, where $y_i = \frac{d_i - \hat{d}_i}{m_d}$, $i = 1, ..., N$, such that $y_i \sim \mathcal{N}(\mu_{y_i}, \sigma_{y_i}^2)$ and $y_a \sim \mathcal{N}(\mu_{y_a}, \sigma_{y_a}^2)$, with $\mu_{y_i} = \frac{\|\boldsymbol{x} - \boldsymbol{a}_i\| - \hat{d}_i}{m_d}$ and $\mu_{y_a} = \frac{\|\boldsymbol{x} - \boldsymbol{a}_a\| + \delta - \hat{d}_a}{m_d}$, $i = 1, ..., N$, $i \neq a$, and $\sigma_{y_i}^2 =  \left(\frac{\sigma}{m_d}\right)^2$, $i = 1, ..., N$. Hence, $e_i$ represents a random variable that follows the folded normal distribution~\cite{Tsagris:2014}.

According to the proposed procedure described in the previous section, the probability of detection is given by
\begin{equation}
P_D = P\left( e_a > \underset{i, \, i \neq a}{\max}\left\{ e_i, \, \tau \right\} \right).
\label{eq:prob_det_gen}
\end{equation}

Even though it is not possible to obtain a solution in a closed form for~\eqref{eq:prob_det_gen}, we can obtain a lower bound and an upper bound on the probability according to the following lemma.
\begin{lemma}
Probability of detection, $P_D$, can be bounded by $LP_D \leq P_D \leq UP_D$, where $LP_D$ and $UP_D$ denote a lower bound and an upper bound on $P_D$, respectively.
\end{lemma}
\begin{proof}
See Appendix~\ref{app:PD}.
\end{proof}

Notice that $P_D$ might be enhanced by simply keeping in mind that the attacker can only enlarge a distance measurement and by following the procedure in Section~\ref{subsec:attack_det} (Case (b)). 

\subsection{Target Localization}
\label{subsec:tar_loc}

After the proposed attacker detection procedure is executed, we turn to the localization problem itself. Our basic idea is to exclude any \emph{corrupted} anchor node which might be detected in the previous step and determine the unknown target location via a bisection procedure, by resorting to \emph{non-corrupted} radio measurements only.

First, we define weights of the form $$w_i = \frac{d_i^{-1}}{\sum_{i\in\mathcal{A}\setminus\mathcal{S}}d_i^{-1}}$$ in order to assign more belief to \emph{nearby} links. Then, according to~\eqref{eq:our_model}, the localization problem can be posed in the form of a weighted squared range approach as
\begin{equation}
\underset{\boldsymbol{x}}{\text{minimize}} \displaystyle\sum_{i\in\mathcal{A}\setminus\mathcal{S}} w_i \left( \| \boldsymbol{x} - \boldsymbol{a}_i \|^2 - d_i^2 \right)^2,
\label{eq:sr-wls_1}
\end{equation}
which, by developing the squared-norm term within the brackets in~\eqref{eq:sr-wls_1} and introducing an auxiliary variable $\alpha = \| \boldsymbol{x} \|^2$, can be written in vector form as
\begin{equation}
\underset{\boldsymbol{y} = [\boldsymbol{x}^T, \alpha]^T}{\text{minimize}} \,\, \left\{ \|\boldsymbol{W}(\boldsymbol{H}\boldsymbol{y} - \boldsymbol{h})\|^2: \boldsymbol{y}^T \boldsymbol{F} \boldsymbol{y} + 2 \boldsymbol{f}^T \boldsymbol{y} = 0 \right\},
\label{eq:sr-wls}
\end{equation}
where $\boldsymbol{W} = \text{diag}(\boldsymbol{w})$, with $\boldsymbol{w} = [\sqrt{w_i}]^T$ and $\text{diag}(\bullet)$ denoting a diagonal matrix where entries on the main diagonal are equal to elements of $\bullet$, and
\begin{equation}
\boldsymbol{H} =
\begin{bmatrix}
\vdots\\
-2 \boldsymbol{a}_{i}^T & 1\\
\vdots
\end{bmatrix},
\,
\boldsymbol{h} =
\begin{bmatrix}
\vdots\\
d_i^2 - \|\boldsymbol{a}_{i}\|^2\\
\vdots
\end{bmatrix},
\nonumber
\end{equation}
\begin{equation}
\boldsymbol{F} =
\begin{bmatrix}
\boldsymbol{I}_2 & \boldsymbol{0}_{2 \times 1}\\
\boldsymbol{0}_{1 \times 2} & 0\\
\end{bmatrix},
\,
\boldsymbol{f} =
\begin{bmatrix}
\boldsymbol{0}_{2 \times 1}\\
-\frac{1}{2}\\
\end{bmatrix},
\nonumber
\end{equation}
with $\boldsymbol{I}_v$ representing the identity matrix of size $v$, and $\boldsymbol{0}_{g \times z}$ denoting the all zero entry matrix of size $g \times z$.

The problem in \eqref{eq:sr-wls} is known as GTRS in the literature~\cite{More:1993, Beck:2008, Tomic_TVT:2015, Tomic_TVT:2016, Tomic_TVT:2020}. Its main particularities are: minimization of a quadratic objective function over a quadratic constraint. Even though GTRS is non-convex in general, it is a monotonically decreasing function over a readily computable interval. Therefore, GTRS is quite convenient for solving via bisection mechanism. We outline the procedure for solving~\eqref{eq:sr-wls} in Appendix~\ref{app:GTRS}.

Finally, to conclude this section, we summarize the generalized version of the proposed method in a universal setting as a pseudo-code in Algorithm~\ref{al:sr-wls}.
\begin{algorithm}[H]\footnotesize
\caption{\hspace*{2mm}Summary of the proposed algorithm}
\begin{algorithmic}[1]
\REQUIRE $\boldsymbol{a}_i, \,\, d_{i,k}, \,\, 1 \leq i \leq N, \,\, 1 \leq k \leq K$

\STATE \textbf{Initialize:} $\mathcal{S} = \varnothing$

{\fontfamily{qcr}\selectfont //Calculate all intersection points}
\STATE $\boldsymbol{p}_{ij} \leftarrow$ \eqref{eq:intersections}

{\fontfamily{qcr}\selectfont //Find circles with no intersections}
\WHILE{$c_i \cap c_j = \varnothing \, \& \, \|\boldsymbol{a}_i - \boldsymbol{l}_i\| > \|\boldsymbol{a}_i - \boldsymbol{l}_j\|, \forall j \in \mathcal{A}$}
	\STATE $\mathcal{S} \leftarrow \mathcal{S} \cup \left\{ i \right\}$
	\STATE $\mathcal{A} \leftarrow \mathcal{A} \setminus \left\{ i \right\}$
	\IF{$|\mathcal{A}| = q+1$}
		\STATE Go to step $19$
	\ENDIF
\ENDWHILE

{\fontfamily{qcr}\selectfont //Obtain initial location estimate}
\STATE $\hat{\boldsymbol{x}}^{(1)} \leftarrow$ \eqref{eq:x_est1}

{\fontfamily{qcr}\selectfont //Estimate attack intensity}
\STATE $\widehat{\delta}_i^{(1)} \leftarrow \frac{\sum_{k=1}^K \left( d_{ik} - \|\hat{\boldsymbol{x}}^{(1)} - \boldsymbol{a}_i\|\right)}{K}$

{\fontfamily{qcr}\selectfont //Compute cost function value}
\STATE $\hat{f}_1 \leftarrow$ \eqref{eq:mle}, using $\hat{\boldsymbol{x}}^{(1)}$ and $\widehat{\delta}_i^{(1)}$

{\fontfamily{qcr}\selectfont //Calculate distance estimates}
\STATE $\hat{d}_i \leftarrow$ \eqref{eq:dist_est}

{\fontfamily{qcr}\selectfont //Calculate relative error}
\STATE $e_i \leftarrow$ \eqref{eq:rel_error}

{\fontfamily{qcr}\selectfont //Detect attackers}
\WHILE{$e_m = \max \{e_1, ..., e_N\} > \tau, m \in \mathcal{A} \, \& \, |\mathcal{A}| > q+1$}
	\STATE $\mathcal{S} \leftarrow \mathcal{S} \cup \left\{ m \right\}$
	\STATE $\mathcal{A} \leftarrow \mathcal{A} \setminus \left\{ m \right\}$
\ENDWHILE

{\fontfamily{qcr}\selectfont //Solve GTRS}
\STATE $\hat{\boldsymbol{y}} \leftarrow$ \eqref{eq:sr-wls}

{\fontfamily{qcr}\selectfont //Obtain updated location estimate}
\STATE $\hat{\boldsymbol{x}}^{(2)} \leftarrow \left[\hat{\boldsymbol{y}}\right]_{1:q}$

{\fontfamily{qcr}\selectfont //Estimate attack intensity}
\STATE $\widehat{\delta}_i^{(2)} \leftarrow \frac{\sum_{k=1}^K \left( d_{ik} - \|\hat{\boldsymbol{x}}^{(2)} - \boldsymbol{a}_i\| \right)}{K}$

{\fontfamily{qcr}\selectfont //Compute cost function value}
\STATE $\hat{f}_2 \leftarrow$ \eqref{eq:mle}, using $\hat{\boldsymbol{x}}^{(2)}$ and $\widehat{\delta}_i^{(2)}$

{\fontfamily{qcr}\selectfont //Obtain final location estimate}
\IF{$f_2 \leq f_1$}
	\STATE $\hat{\boldsymbol{x}} \leftarrow \hat{\boldsymbol{x}}^{(2)}$
\ELSE
	\STATE $\hat{\boldsymbol{x}} \leftarrow \hat{\boldsymbol{x}}^{(1)}$
\ENDIF
\end{algorithmic}
\label{al:sr-wls}
\end{algorithm}



\section{Simulation Results}
\label{sec:results}

This section presents a set of numerical results in order to assess the performance of the proposed algorithm, both in terms of localization accuracy and success in attacker detection. In the simulations presented here, all nodes were randomly deployed $N_{D} = 500$ times within a $20 \times 20 \, \text{m}^2$ area. The radio measurements were acquired according to~\eqref{eq:our_model}. Moreover, each of the $N$ anchor nodes was considered as corrupted $N_{C} = 100$ times for each node deployment, with its radio measurement generated by following~\eqref{eq:corrup_meas}. Unless stated otherwise, a single corrupted anchor node is considered in the following simulations and the number of measurement samples is set to $K=10$. The main metric for localization accuracy assessment is the root mean squared error (RMSE), defined as $\text{RMSE} = \sqrt{\sum_{m=1}^{M_c} \frac{\|\boldsymbol{x}_{m} - \widehat{\boldsymbol{x}}_{m}\|^2}{M_c}},$
where $\widehat{\boldsymbol{x}}_{m}$ is the estimate of the true target location, $\boldsymbol{x}_{m}$, in the $m$-th Monte Carlo, $M_c = N_{D} \times N_{C} \times N$, run.

The section is organized in two parts: the first one analyses the influence of $\tau$ on the performance of the proposed method, while the second one compares its performance against the state-of-the-art methods.

\subsection{Analysis of the Choice of $\tau$}
\label{subsec:T_analysis}

Figs. \ref{fig:RMSE_vs_Cond_N4} and \ref{fig:RMSE_vs_Cond_N5} illustrate the RMSE (m) versus $\delta$ (m) performance of the proposed algorithm for different values of $\tau$, when $N=4$ and $N=5$, respectively. It is worth mentioning that the case where $N=4$ corresponds to the bare minimum of the number of anchor nodes required for secure localization in 2-dimensional networks in the presence of an attacker. The proposed algorithm is compared against its two counterparts: one where perfect detection of the attacker is always available (which can be seen as a lower bound on the performance of the proposed algorithm), and the other one where no detection is performed at all (using all available information from the anchor nodes). The figures exhibit various interesting information; lets start by analyzing the proposed method in terms of the choices of $\tau$. As mentioned earlier, this parameter is intended to serve as a threshold in order to \emph{distinguish} between noise and attack. Interestingly, both figures show that for low attack intensity (e.g., $\delta \leq 2$ m), the RMSE performance somewhat betters for greater choice of $\tau$. This can be explained by the fact that large values of $\tau$ make the proposed algorithm somewhat insensible to attacks of relatively low intensity, and force it to use all available information (even the corrupted one) most of the time. Nevertheless, it turns out that this is beneficial when $\delta$ is relatively small, because it does not differ much from noise, and thus, since $N$ is low, it is actually better to take advantage of all available information than to discard any of it (even if it is corrupted). This can also be confirmed by looking at the two counterparts of the proposed algorithm, where the figures show that not performing any detection scheme and using all available information is better than employing the non-corrupted information only when $\delta$ is very low. As the attack intensity grows, the situation turns around slowly, and somewhat better localization results are obtained for lower values of $\tau$. Therefore, although it seems that there is no fixed $\tau$ that is the best for all considered range of $\delta$, it is important to note that the proposed method shows fairly good robustness to the choice of $\tau$, i.e., its performance does not suffer dramatic deterioration depending on the choice of $\tau$. Obviously, some tuning can be accomplished, but the difference in the performance is not drastic. Intuitively, the value of $\tau$ can be interpreted as the tolerance for the difference between the measured and the estimated distances in terms of the median of the measured ones. Therefore, choosing an excessive value for $\tau$ results in exorbitant tolerance, i.e., the sensibility of the proposed procedure gets numbed and might not detect some more intense attacks. To us, it seems that the choice $\tau \in \left[30,60\right]\%$ is reasonable, since it does not make the method highly sensitive nor too numb to different attack intensities. Now, if we compare these results (e.g., $\tau=30\%$) against the results where no detection is performed, we will see the maximum performance loss is roughly 1 m if one employs the proposed algorithm for $\delta \leq 2$ m and $N=4$, while for more intense attacks the proposed method brings benefits, and can reduce the localization error as high as 5 and 6 meters for $N = 4$ and $N = 5$, respectively. This surely justifies the use of the proposed method, since even in the benign scenario, i.e., when $\delta = 0$ m, there is no significant loss in the localization performance, while it can bring great advantage in malign environments with high attack intensity.
\begin{figure}
\begin{subfigure}{.5\textwidth}
\hspace*{0mm}\includegraphics[width=\textwidth]{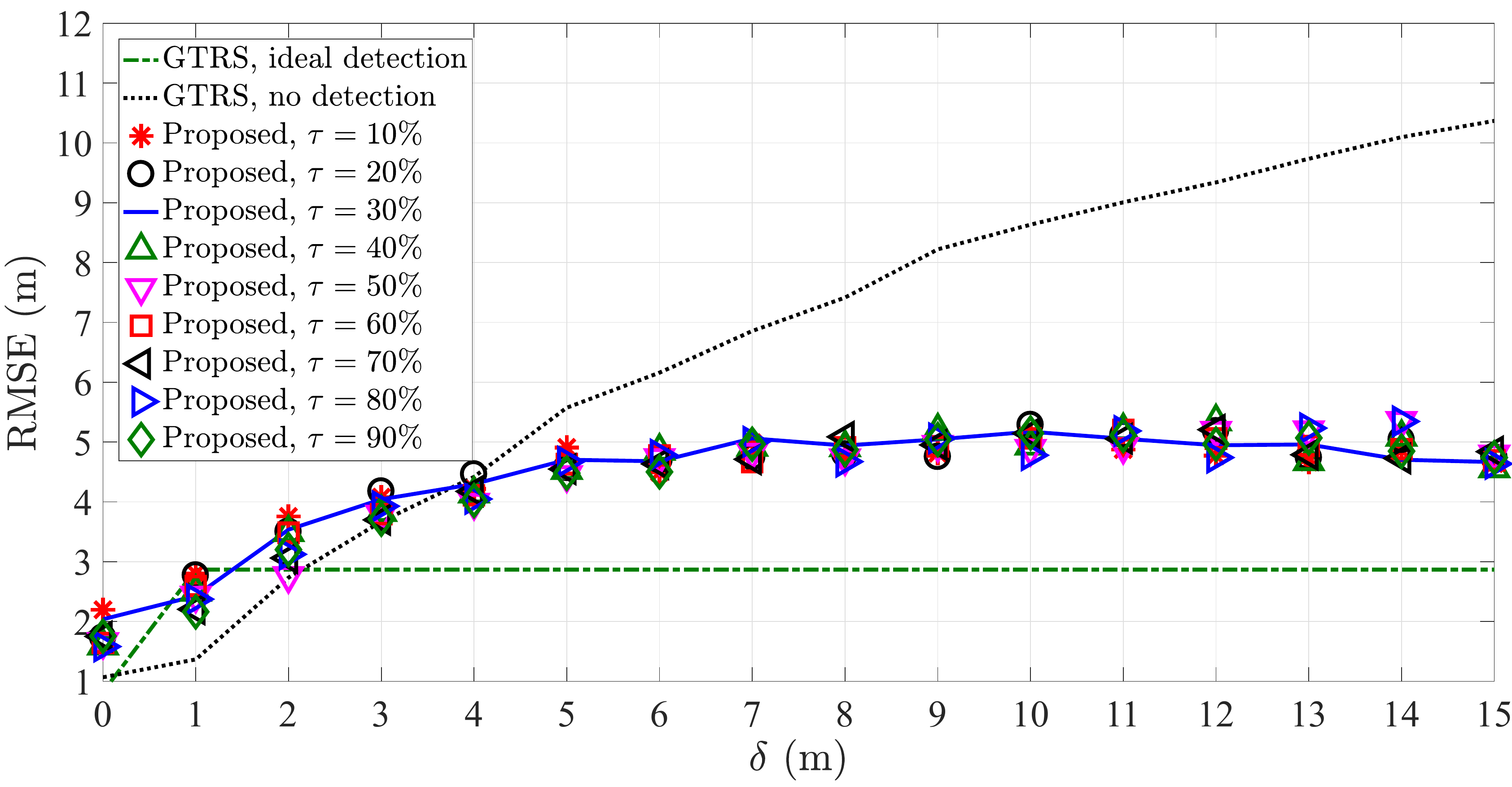}
\caption{$N=4$}
\label{fig:RMSE_vs_Cond_N4}
\end{subfigure}
\vspace*{3mm}
\begin{subfigure}{.5\textwidth}
\hspace*{0mm}\includegraphics[width=\textwidth]{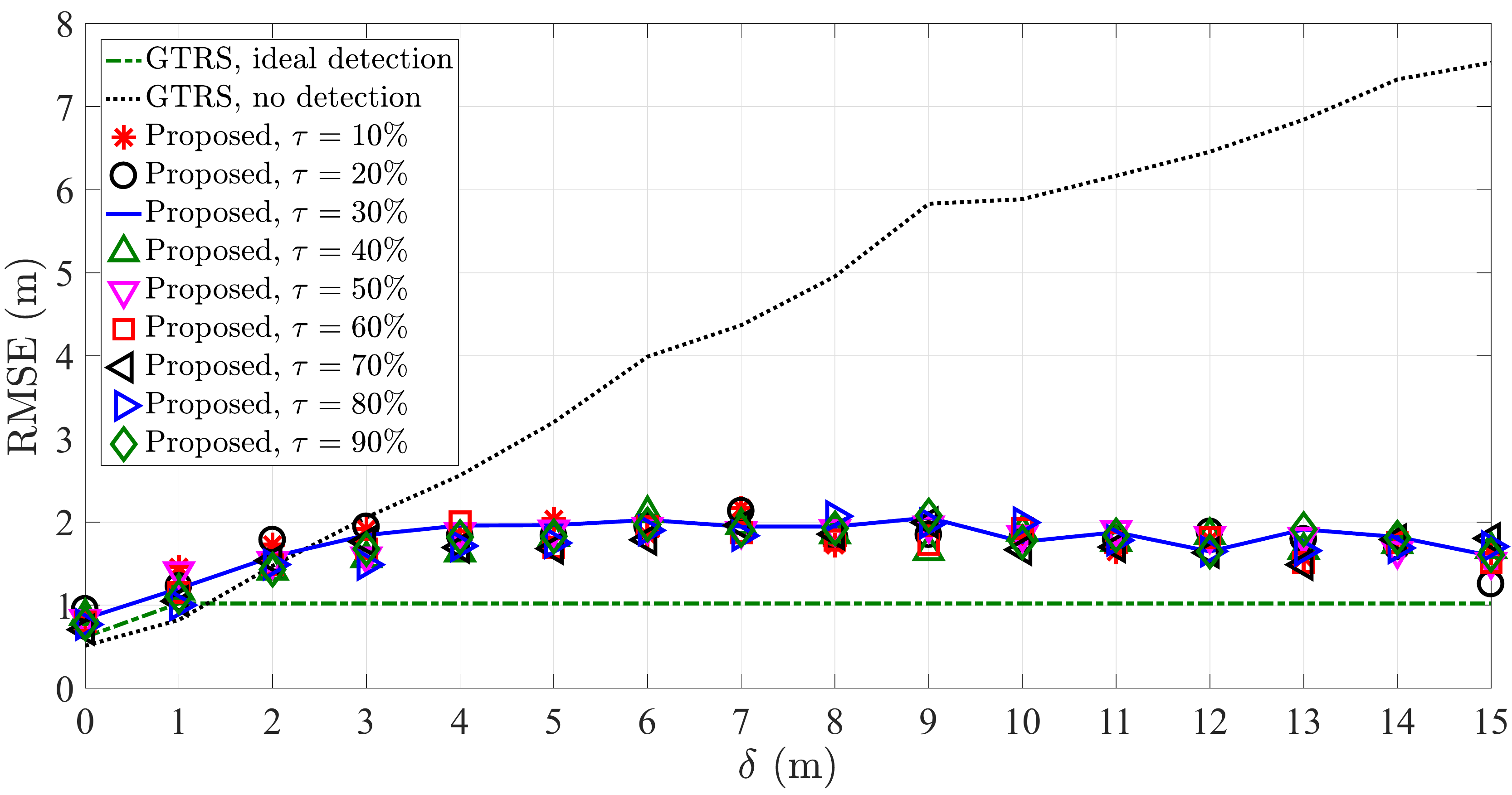}
\caption{$N=5$}
\label{fig:RMSE_vs_Cond_N5}
\end{subfigure}
\vspace*{-3mm}
\caption{RMSE (m) versus $\delta$ (m) illustration, for different choices of $\tau$, when $\sigma = 1$ m.}
\label{fig:RMSE_vs_Cond}
\end{figure}

Figs. \ref{fig:Detection_vs_N4} and \ref{fig:Detection_vs_N5} illustrate the performance of the proposed algorithm in terms of success in attacker detection ($\%$) for different values of $\delta$ (m), when $\tau=30\%$, and $N=4$ and $N=5$, respectively. As desired, one can see from the figures that for low attack intensity (e.g., $\delta \leq 2$ m), the proposed algorithm does not detect any attack in most cases, which means that it treats low-intensity attacks as noise and intentionally uses all available measurements to enhance its localization accuracy. However, as the attack intensity increases, the success in correct detection of the attacker also increases. This behavior is anticipated, because when $\delta$ grows, the attack gets more accentuated which makes it more difficult for the malicious node to hide its attack within the measurement noise. One can see that the success in attacker detection tends to intensify with the increase of $\delta$, and it goes as far as over $95\%$ and over $99\%$ for $N=4$ and $N=5$ respectively. Even though the success in attacker detection is close to being perfect for intense attacks, each error committed in this case has greater consequences in terms of the localization error, which is why even in the case of high detection success, there is still room for improvement of the proposed algorithm in terms of the localization accuracy (e.g., please see Figs. \ref{fig:RMSE_vs_Cond_N4} and \ref{fig:RMSE_vs_Cond_N5} for $\delta = 15$ m).
\begin{figure}
\begin{subfigure}{.5\textwidth}
\hspace*{0mm}\includegraphics[width=\textwidth]{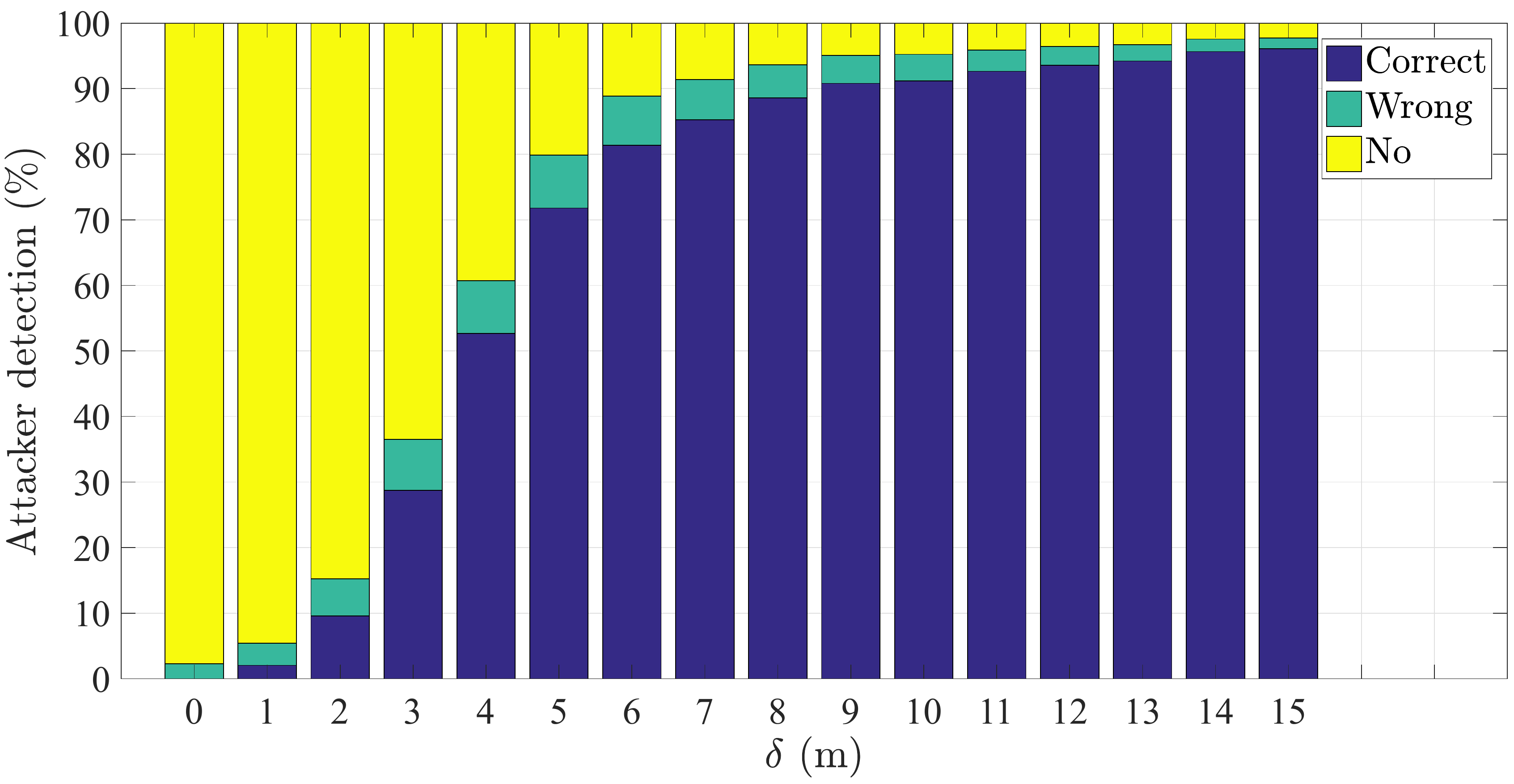}
\caption{$N=4$}
\label{fig:Detection_vs_N4}
\end{subfigure}
\vspace*{3mm}
\begin{subfigure}{.5\textwidth}
\hspace*{0mm}\includegraphics[width=\textwidth]{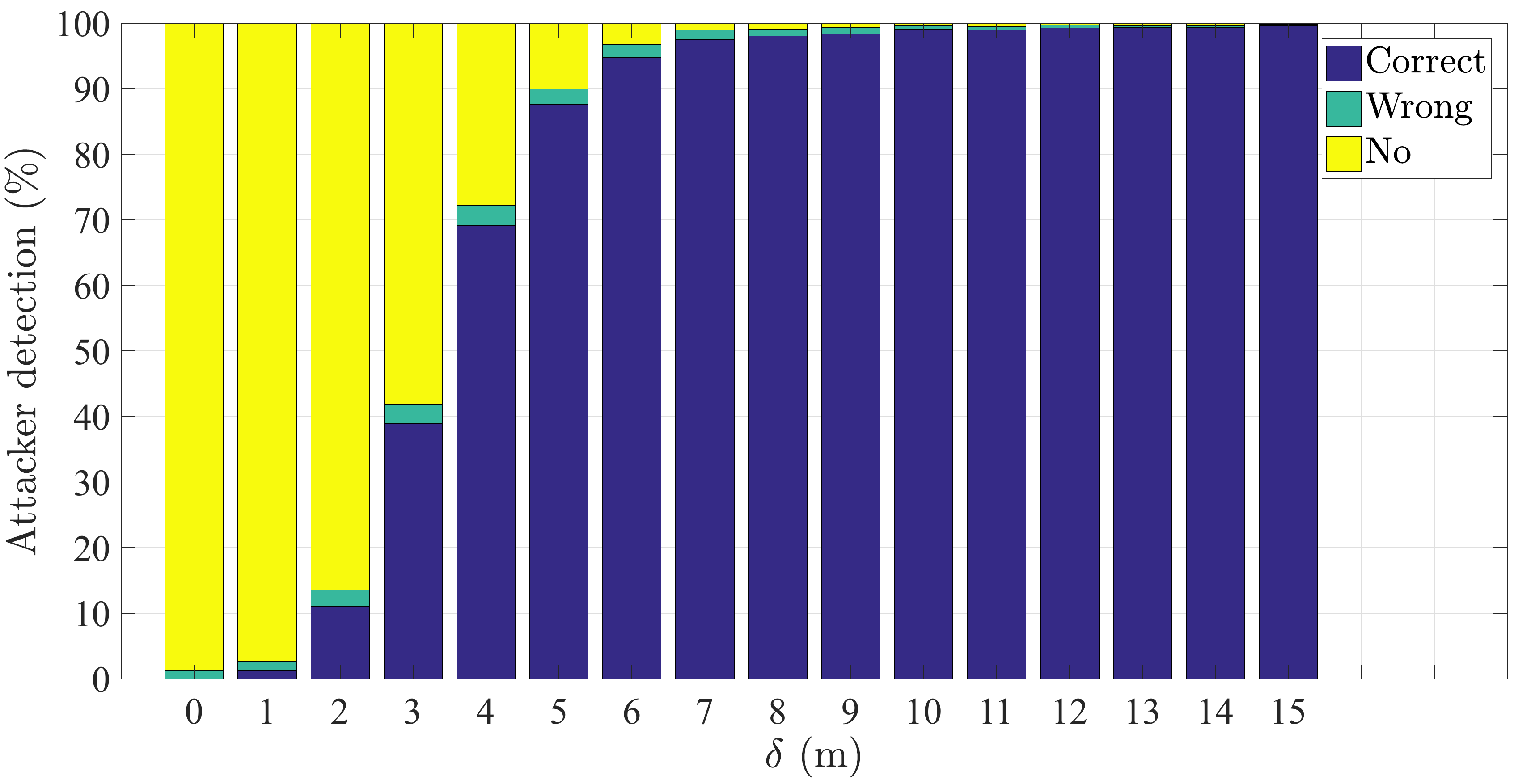}
\caption{$N=5$}
\label{fig:Detection_vs_N5}
\end{subfigure}
\vspace*{-3mm}
\caption{Attacker detection ($\%$) versus $\delta$ (m) illustration, when $\tau = 30\%$ and $\sigma = 1$ m.}
\label{fig:Detection_vs_N}
\end{figure}

Figs. \ref{fig:PD_vs_delta_for_N4} and \ref{fig:PD_vs_delta_for_N5} illustrate the probability of detection versus $\delta$ (m) comparison, for $N=4$ and $N=5$ respectively. The results in the figures show a good match between the simulations and theory, and one can see that the probability of detection of the proposed method is very close to the upper bound provided by~\eqref{eq:UPD}. Although according to \eqref{eq:LPD1}-\eqref{eq:UPD}, one should be inclined to choose a lower threshold in order to enhance $P_D$ (which is also intuitive from the procedure described in Section~\ref{subsec:prob_det}; e.g., see the line $15$ in Algorithm~\ref{al:sr-wls}), the attacker detection problem should not be considered completely independent from the localization problem, since there is some trade off between the two, at least for relatively low attack intensities (e.g., please see Figs. \ref{fig:RMSE_vs_Cond_N4} and \ref{fig:RMSE_vs_Cond_N5} for $\delta \leq 3$ m).
\begin{figure}
\begin{subfigure}{.5\textwidth}
\hspace*{0mm}\includegraphics[width=\textwidth]{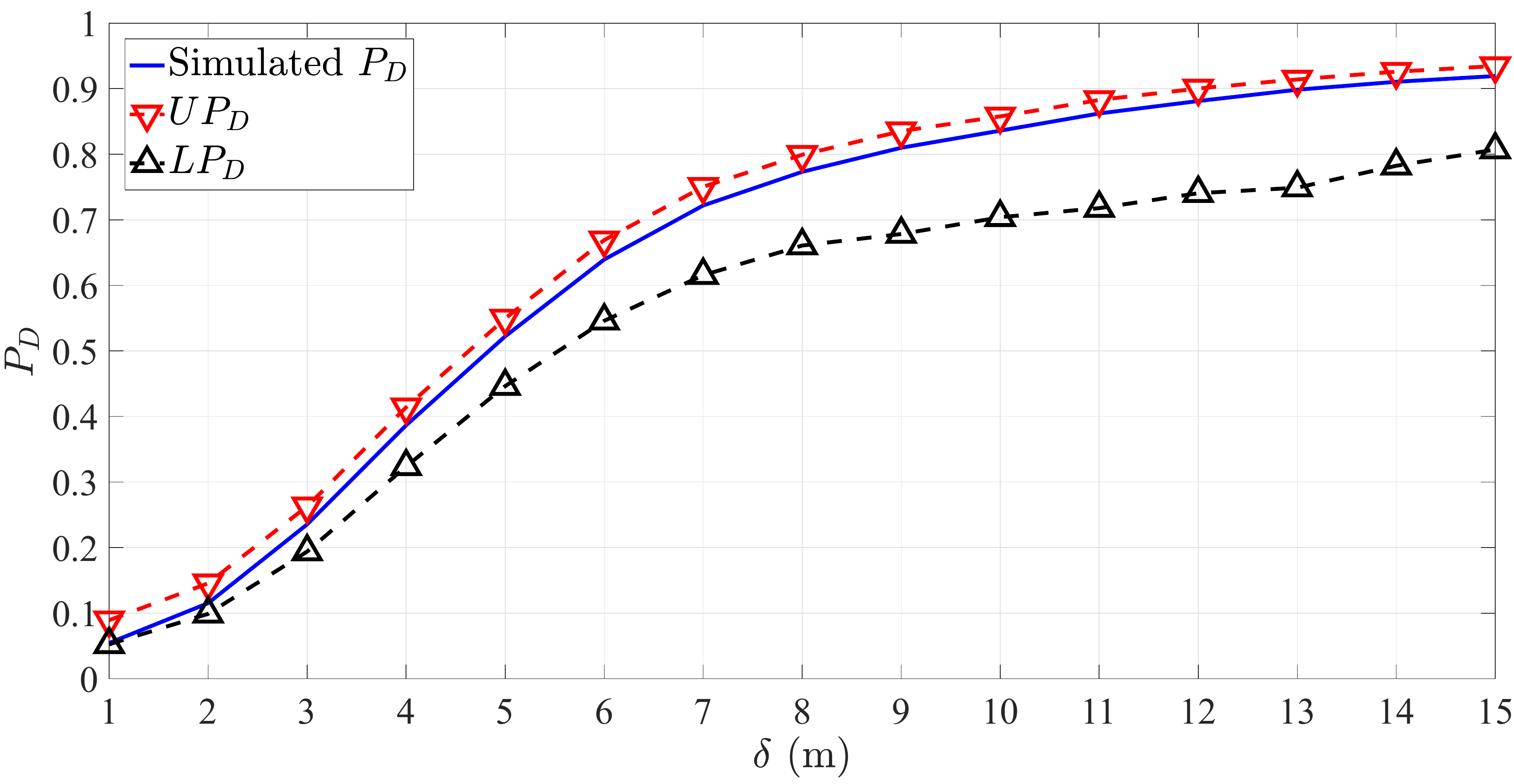}
\caption{$N=4$}
\label{fig:PD_vs_delta_for_N4}
\end{subfigure}
\vspace*{3mm}
\begin{subfigure}{.5\textwidth}
\hspace*{0mm}\includegraphics[width=\textwidth]{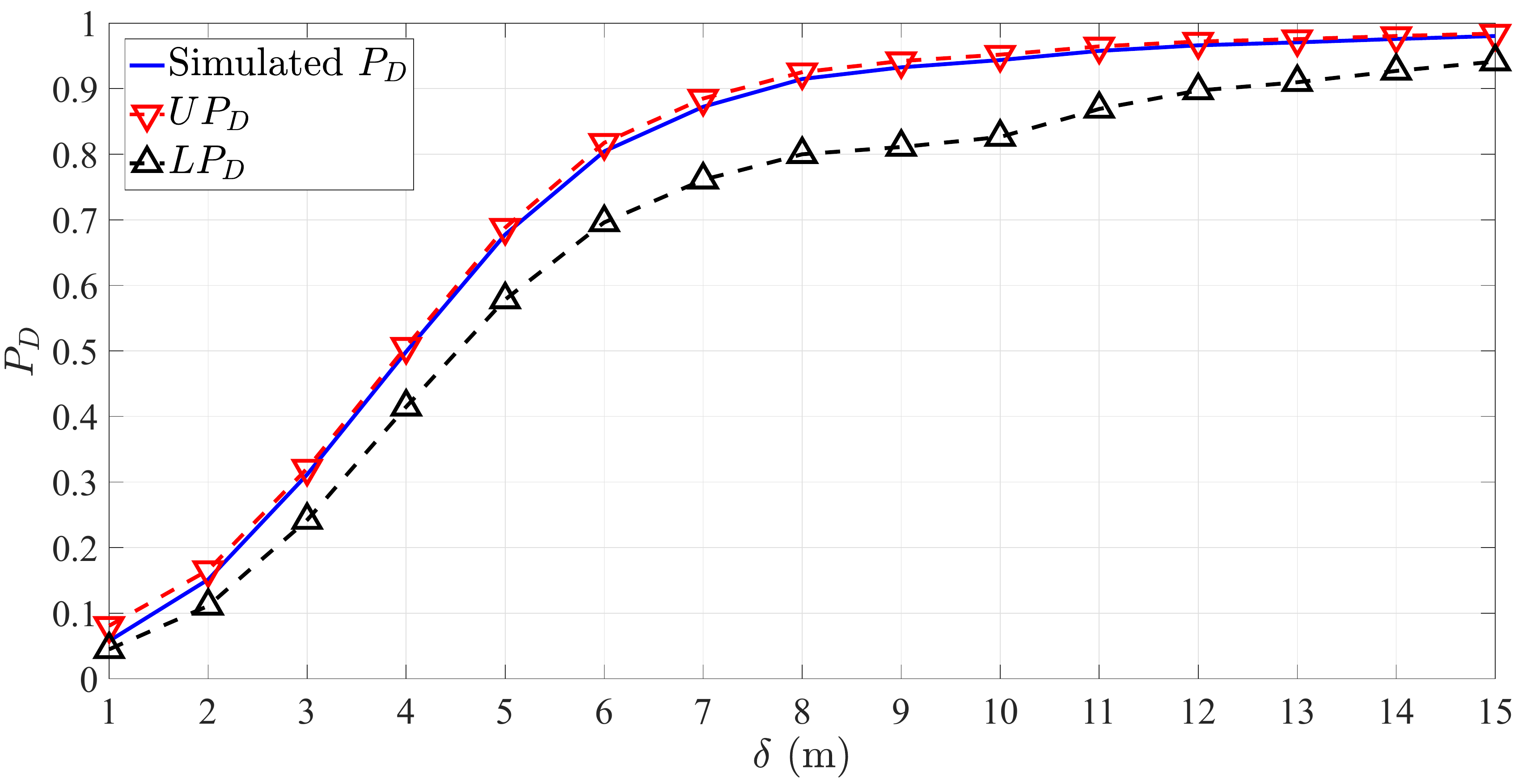}
\caption{$N=5$}
\label{fig:PD_vs_delta_for_N5}
\end{subfigure}
\vspace*{-3mm}
\caption{$P_D$ versus $\delta$ (m), when $\tau = 30\%$, $K=1$ and $\sigma = 1$ m.}
\label{fig:PD_vs_delta_and_N}
\end{figure}

Let us summarize the main findings thus far: 1) we saw that an obvious and efficient way to reduce the negative effect of a malicious attacker is to simply increase the number of (non-corrupted) anchor nodes in the network. However, this is not always feasible nor is the main goal in many applications; 2) even though theory indicates that one should set a low $\tau$ to enhance $P_D$, detection performance alone should not be considered of primary interest when the attack intensity is fairly low, since treating these attacks as noise (rather than simply disregarding them) can lead to improvement in the RMSE performance; 3) the choice of $\tau$ for the proposed scheme plays a tuning role only, and it does not have a crucial influence on its localization performance.

\subsection{Comparison with Existing Methods}
\label{subsec:results}

In this section, the performance of the proposed algorithm is compared with the WLS method in \cite{Mukhopadhyay:2018} and the VBL method in \cite{Li:2020}, which are considered as the state-of-the-art methods for secure localization and link identification in mixed LOS/NLOS environments, respectively. It is worth mentioning that the original implementation of WLS does not include any detection scheme, but rather tries to minimize the negative effect of the corrupted anchor node by employing specially designed weights. Thus, we implemented a \emph{classical} detection scheme, GLRT, to WLS in all results presented here. The main details about GLRT detection are given in Appendix \ref{app:GLRT}.

Let us start by studying the computational complexity of the algorithms. Assuming that $N_p$ and $\eta$ denote respectively the number of particles drawn in the importance sampling part for expextation derivations and the number of iteration of the VBL algorithm, and that $B_{\text{max}}$ stands for the maximum number of iterations in the bisection procedure of the proposed algorithm, Table~\ref{tab:complexity} outlines the worst case computational complexities. From the table, one can see that all three algorithms have linear computational complexity in $N$. However, the proposed solution has somewhat increased complexity in comparison with WLS due to the use of bisection procedure, while the complexity of VBL is dominated by the number of particles and the number of iterations. Hence, in terms of computational complexity, the proposed algorithm represents a fair alternative.
\begin{table}\footnotesize
\caption{Complexity Analysis of the Considered Algorithms}
\vspace*{-2mm}
  \begin{center}
	\begin{tabular}{|c|c|c|c|}
	\hline
	\textbf{Algorithm} & Complexity\\ \hline \hline
	WLS in~\cite{Mukhopadhyay:2018} & $\mathcal{O}\left( N \right)$\\ \hline
	VBL in~\cite{Li:2020} & $\mathcal{O}\left( N N_p^2\eta \right)$\\ \hline
	Proposed in Algorithm~\ref{al:sr-wls} & $\mathcal{O}\left( B_{\text{max}N} \right)$\\ \hline
	\end{tabular}
   \end{center}
\label{tab:complexity}
\end{table}

Figs.~\ref{fig:RMSE_vs_delta_N4} and~\ref{fig:RMSE_vs_delta_N5} illustrate the RMSE (m) versus $\delta$ (m) comparison of the proposed method for $\tau = 30\%$, WLS, and VBL, when $\sigma=1$ m, for $N=4$ and $N=5$ respectively. The figures clearly illustrate the superiority of the proposed method over the existing ones for practically all values of $\delta$. The figures clearly show that the existing methods cannot handle the increase in the attack intensity of the attacker, since their localization error practically grows monotonically with $\delta$, whereas the proposed estimator shows that there is a critical point about $5 \leq \delta \leq 7$ m, after which the attacker cannot deteriorate its performance.
\begin{figure}
\begin{subfigure}{.5\textwidth}
\hspace*{0mm}\includegraphics[width=\textwidth]{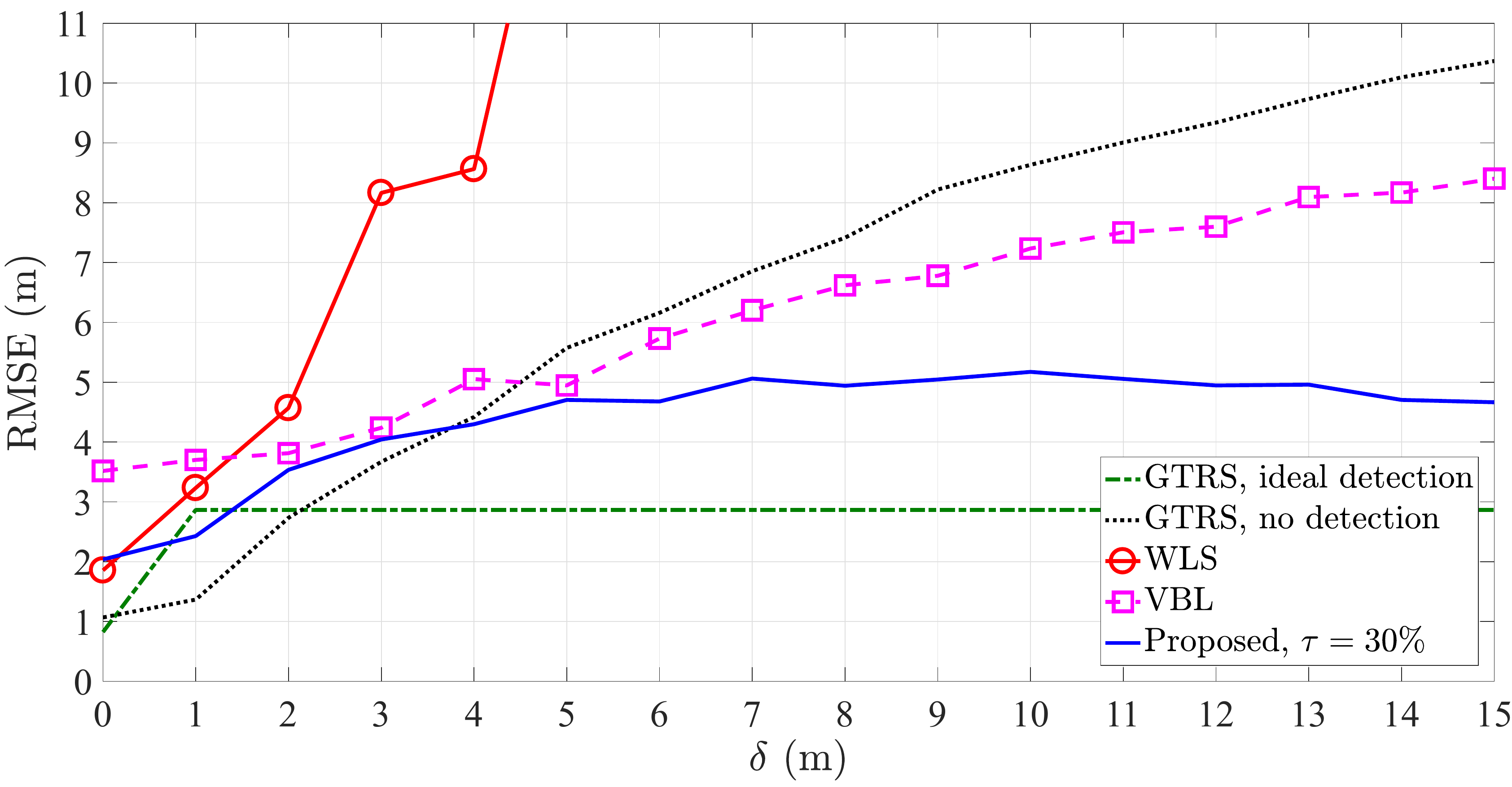}
\caption{$N=4$}
\label{fig:RMSE_vs_delta_N4}
\end{subfigure}
\vspace*{3mm}
\begin{subfigure}{.5\textwidth}
\hspace*{0mm}\includegraphics[width=\textwidth]{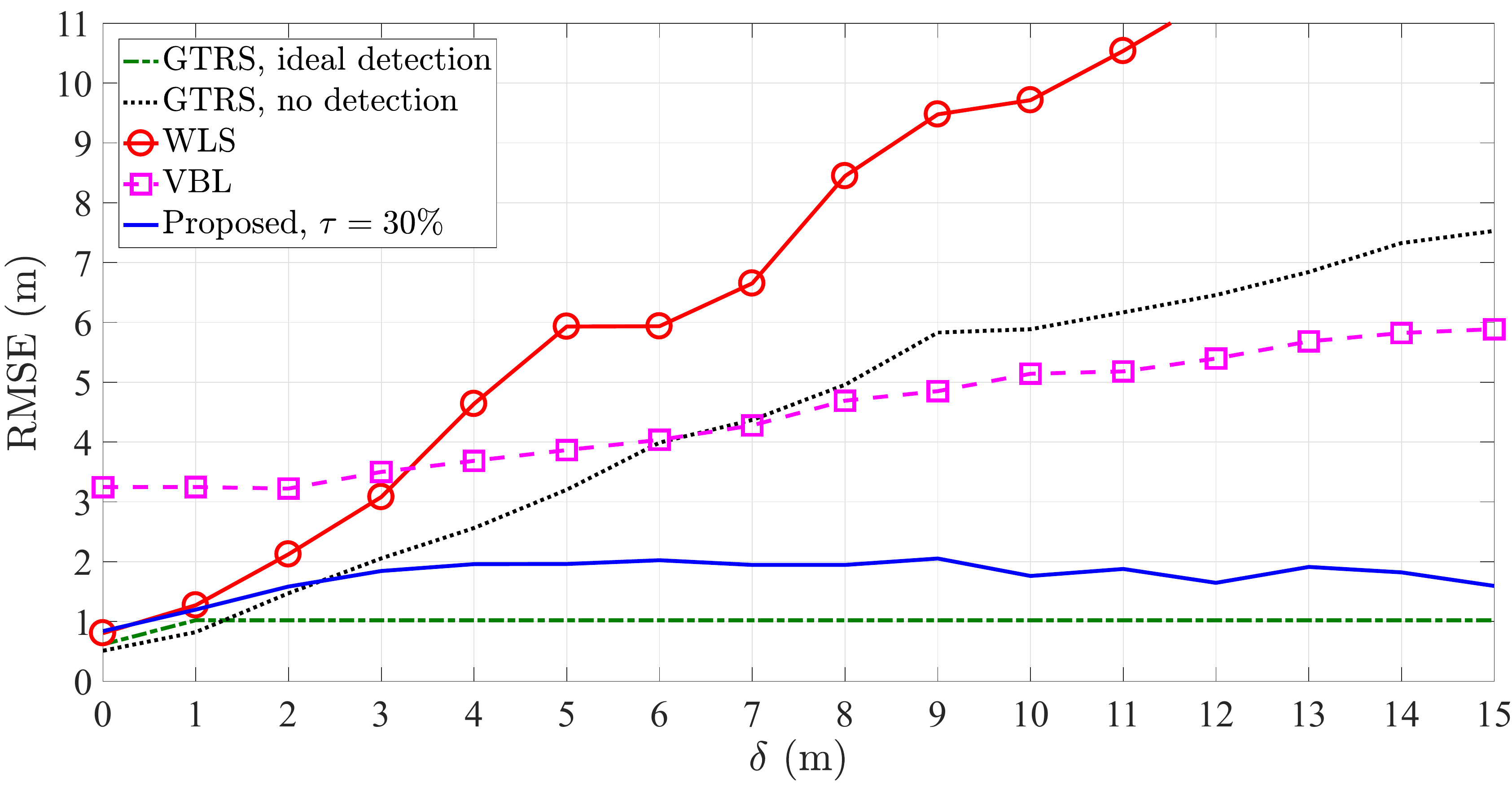}
\caption{$N=5$}
\label{fig:RMSE_vs_delta_N5}
\end{subfigure}
\vspace*{-3mm}
\caption{RMSE (m) versus $\delta$ (m) illustration, $\sigma = 1$ m.}
\label{fig:RMSE_vs_delta}
\end{figure}

Figs.~\ref{fig:PD_vs_delta_N4_soa} and~\ref{fig:PD_vs_delta_N5_soa} illustrate the $P_D$ versus $\delta$ (m) comparison of the proposed method for $\tau = 30\%$, WLS, and VBL, when $\sigma=1$ m, for $N=4$ and $N=5$ respectively. The figures show superior detection performance of the proposed method over the existing ones for $\delta \geq 6$ m. As desired, for low values of $\delta$ the proposed scheme trades the $P_D$ performance for enhanced localization accuracy, as explained in Section \ref{subsec:T_analysis}.
\begin{figure}
\begin{subfigure}{.5\textwidth}
\hspace*{0mm}\includegraphics[width=\textwidth]{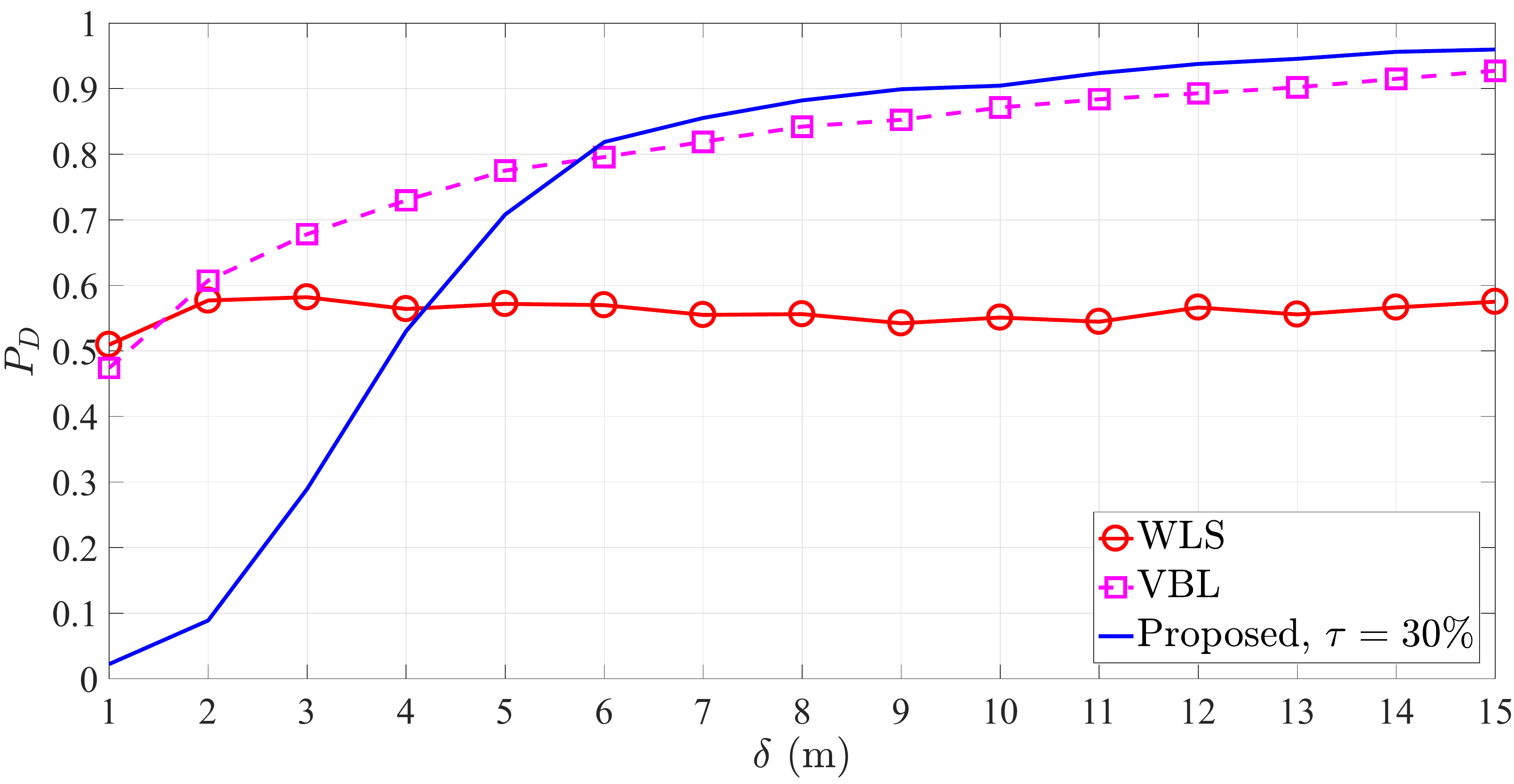}
\caption{$N=4$}
\label{fig:PD_vs_delta_N4_soa}
\end{subfigure}
\vspace*{3mm}
\begin{subfigure}{.5\textwidth}
\hspace*{0mm}\includegraphics[width=\textwidth]{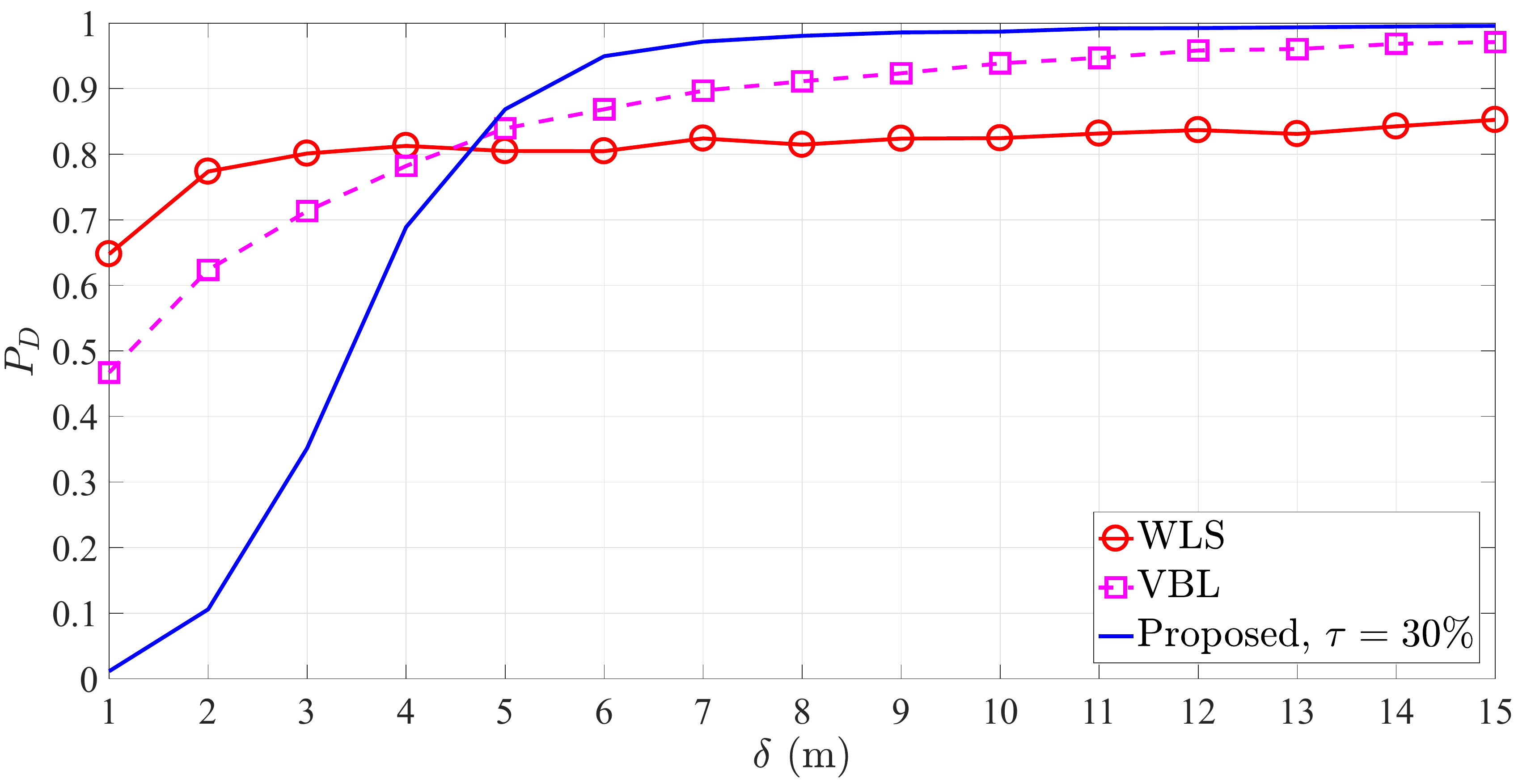}
\caption{$N=5$}
\label{fig:PD_vs_delta_N5_soa}
\end{subfigure}
\vspace*{-3mm}
\caption{$P_D$ versus $\delta$ (m), when $T = 30\%$ and $\sigma = 1$ m.}
\label{fig:PD_vs_delta_soa}
\end{figure}

Lastly, Figs. \ref{fig:RMSE_vs_delta_2_att} and \ref{fig:PD_vs_delta_2_att} illustrate the RMSE (m) and $P_D$ (i.e, the probability of false alarm, $P_{FA}$) versus $\delta$ (m) comparisons respectively, for $N=6$, when two attackers are present at any moment. In this setting, all possible pairs of anchor nodes were considered as corrupted, $N_C$ number of times. It is worth mentioning that, according to the procedure outlined in Algorithm \ref{al:sr-wls}, it is possible to \emph{detect} up to three \emph{attackers}, since at least three \emph{honest} anchor nodes are needed in order to solve the localization problem in a 2-dimensional space. Once again, the proposed solution exhibits superior RMSE performance by far, while guaranteeing at the same time the highest $P_D$ and lowest $P_{FA}$.
\begin{figure}
\begin{subfigure}{.5\textwidth}
\hspace*{0mm}\includegraphics[width=\textwidth]{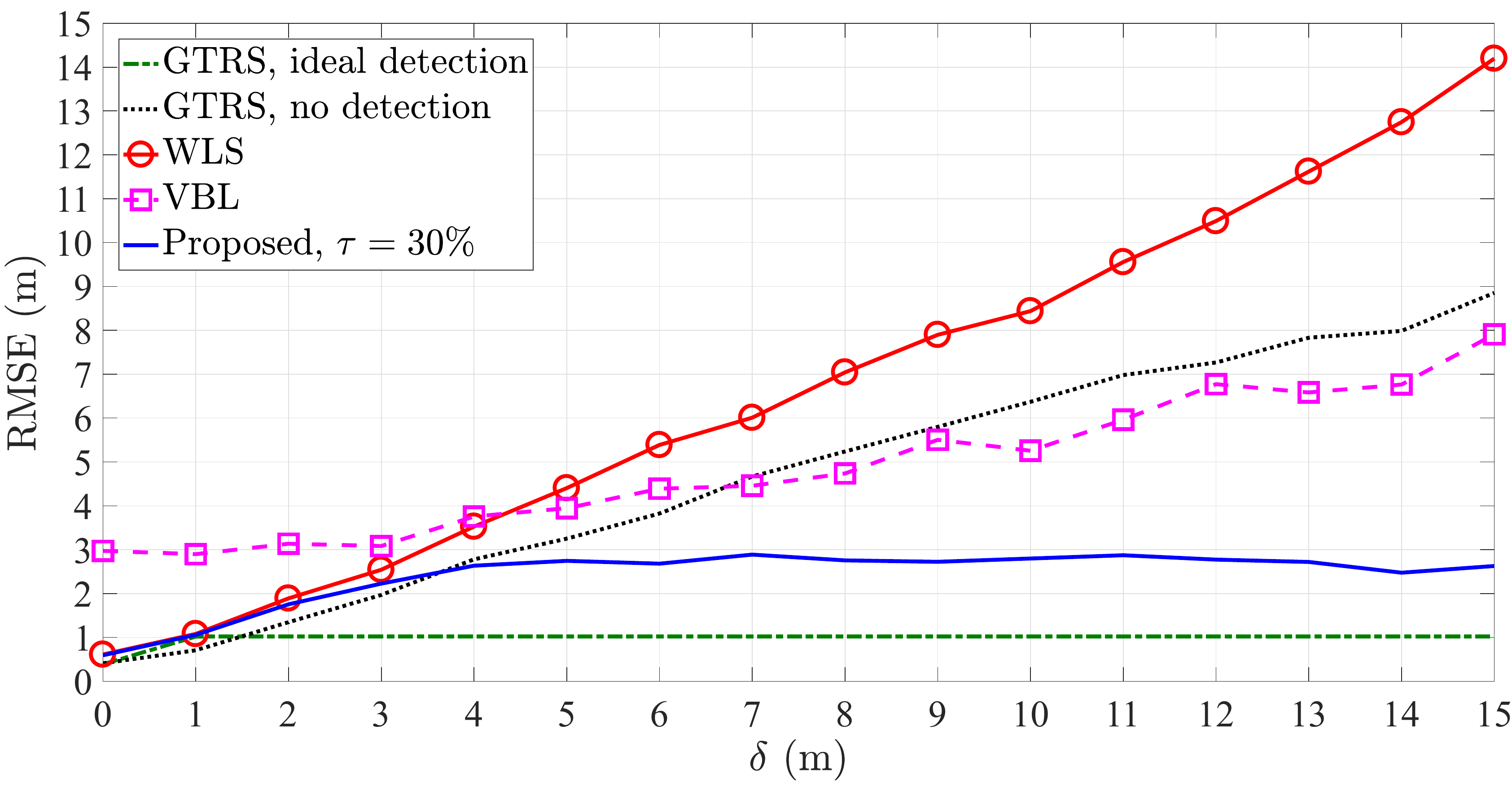}
\caption{RMSE (m) versus $\delta$ (m) comparison}
\label{fig:RMSE_vs_delta_2_att}
\end{subfigure}
\vspace*{3mm}
\begin{subfigure}{.5\textwidth}
\hspace*{0mm}\includegraphics[width=\textwidth]{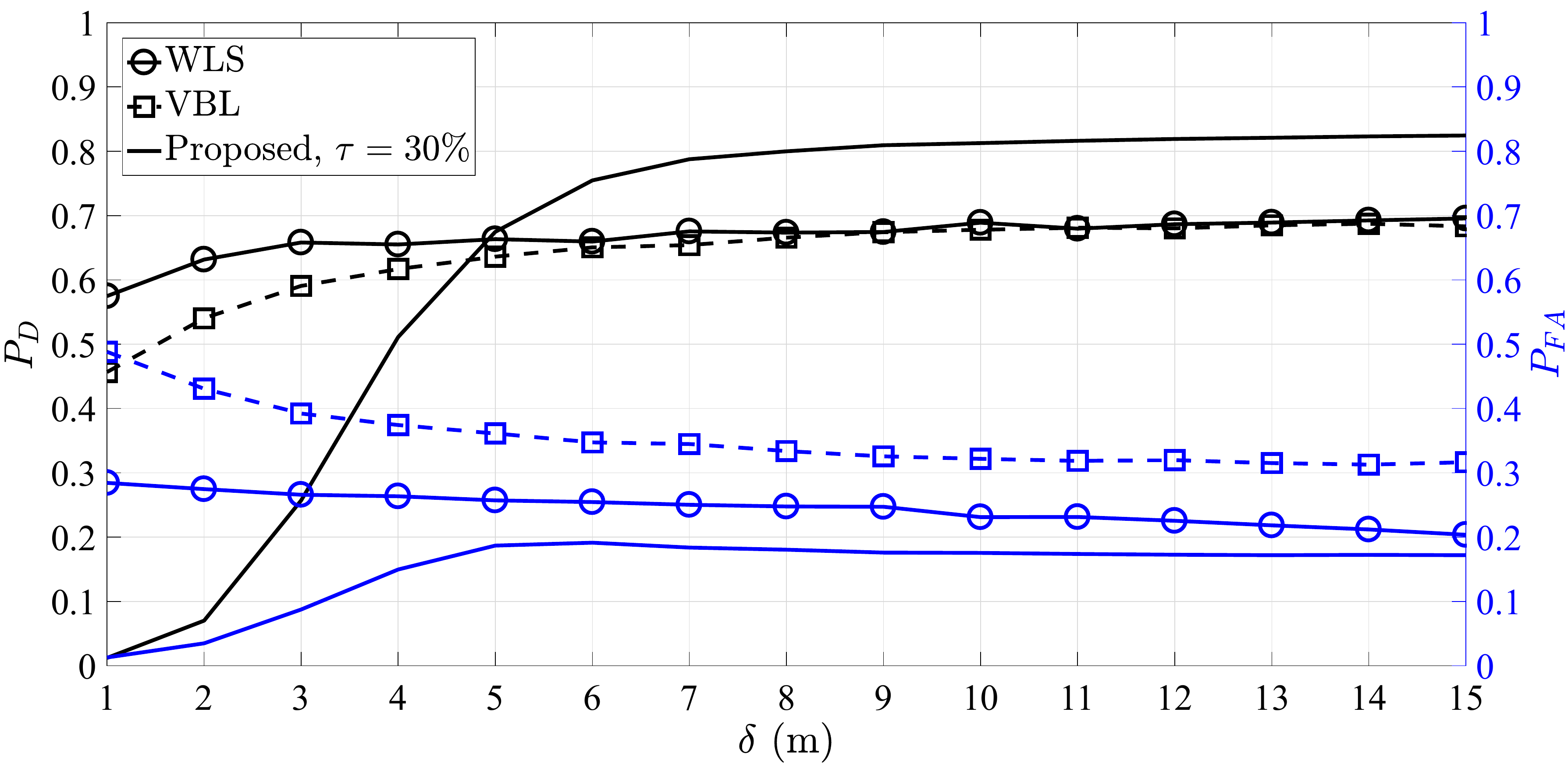}
\caption{$P_D$/$P_{FA}$ versus $\delta$ (m) comparison}
\label{fig:PD_vs_delta_2_att}
\end{subfigure}
\vspace*{-3mm}
\caption{Performance comparison in the presence of two attackers, when $N=6$ and $\sigma = 1$ m.}
\label{fig:RMSE_and_PD_vs_delta_2_att}
\end{figure}



\section{Conclusions}
\label{sec:conclusions}

This work presented a novel approach for secure target localization in randomly deployed wireless networks in the presence of malicious adversaries. This is an important problem since strengthening the security of current non-secure systems or generalization of existing secure localization systems to ad hoc scenarios will enable additional reliable safety parameter (location) to be employed for digital interactions in more general contexts (such as social media, health monitoring or surveillance systems). The proposed algorithm is based on TW-TOA measurements, where external attackers are considered to corrupt (spoof) some of the anchor nodes' measurements, in the sense that they reported enlarged distance measurements to the target. The proposed algorithm can be broken down into three main steps: 1) clustering, in which an initial estimation of the target location was obtained by using WCM, 2) attacker detection, in which the initial estimation is used to detect attackers via threshold-based keying of the relative error between the measured and estimated distances, and 3) localization, in which the localization problem is solved by converting it into a GTRS and solving it by means of a bisection procedure. The proposed method was assessed through a set of simulation results, where it showed promising results from both localization accuracy and success in attacker detection perspectives. Although the new method exhibited good RMSE performance in the considered settings, there is still room for further improvement, which is clear from the performance margin between itself and the employed lower bound obtained by using the proposed localization estimator when perfect attacker detection is available. Therefore, this work represents our first step towards secure localization in randomly deployed networks, with our ultimate goal being achievement of the performance of localization algorithms in benign environments. Nonetheless, to the best of authors' knowledge, this is the first work that treats the secure localization problem in a conceptually novel approach by unifying the localization and attacker detection problems, instead of treating them separately.

A possibly interesting direction for future research might be adaptation of the proposed algorithm (or development of novel ones) for the localization problem based on RSS measurements, since they are widely available in various devices. Also, generalization of the proposed scheme to the case where multiple coordinated attackers are present in the network may be of interest. Finally, location attacks, in which attackers lie about their true locations might be of interest in some practical applications as well.



\section*{Acknowledgements}
The authors would like to express their sincere gratitude to Prof. Brook Luers from the University of Michigan and Prof. Dilip Sarwate from the University of Illinois for their valuable posts on statistical problems publicly available online, which helped us derive $LPD1$. This work was partially supported by Funda\c{c}\~{a}o para a Ci\^{e}ncia e a Tecnologia under Projects UIDB/04111/2020 and foRESTER PCIF/SSI/0102/2017.



\begin{appendices}

\section{Calculation of the Probability of Detection}
\label{app:PD}

A lower bound on the probability of detection in~\eqref{eq:prob_det_gen} can be calculated as follows.
\begin{equation}
\begin{array}{l}
P\left( e_a > \underset{i: \, i \neq a}{\max}\left\{ e_i, \, \tau \right\} \right) =\\
1 - P\left( \abs{y_a} \leq \underset{i, \, i \neq a}{\max}\left\{ \abs{y_i}, \, \tau \right\} \right) =\\
1 - P\left( \bigcup\limits_{i: \, i \neq a} \abs{y_a} \leq \abs{y_i} \, \cup \, \abs{y_a} \leq \tau \right) \geq\\
1 - \left( \displaystyle\sum_{i: \, i\neq a} P\left( \, \abs{y_a} \leq \abs{y_i} \right) + P\left( \, \abs{y_a} \leq \tau \right) \right).
\end{array}
\label{eq:PD_der1}
\end{equation}

The last factor in~\eqref{eq:PD_der1} can be calculated~\cite{Tsagris:2014} as
\begin{equation}
\begin{array}{l}
P(e_a \leq \tau) = P\left( \, \abs{y_a} \leq \tau \right) =\\
\frac{1}{2}\left[ \text{erf}\left( \frac{\tau+\mu_{y_a}}{\sigma_{y_a}\sqrt{2}} \right) + \text{erf}\left( \frac{\tau-\mu_{y_a}}{\sigma_{y_a}\sqrt{2}} \right) \right] =\\
1- \left( Q\left( \frac{\tau+\mu_{y_a}}{\sigma_{y_a}} \right) + Q\left( \frac{\tau-\mu_{y_a}}{\sigma_{y_a}} \right) \right),
\end{array}
\nonumber
\end{equation}
where $\text{erf}(z) = \frac{2}{\sqrt{\pi}} \int_{0}^z e^{-t^2}dt$ is the error function and $Q(\bullet)$ represents the $Q$-function.

To find the remaining factors in~\eqref{eq:PD_der1}, one needs to solve the following integral
\begin{equation}
P(\,\abs{y_a} - \abs{y_i} < 0) = \iint\limits_{R} p_{y_ay_i} dy_ady_i,
\label{eq:orig_int}
\end{equation}
where $p_{y_ay_i}$ is the joint probability density function of $y_a$ and $y_i$, and the region $R = \left\{ [y_a, y_i]^T: \, \abs{y_a} < \abs{y_i} \right\}$ as illustrated on the left-hand side in Fig.~\ref{fig:prob_det}.
\begin{figure}
\centering
\hspace*{-0mm}\includegraphics[width=\linewidth]{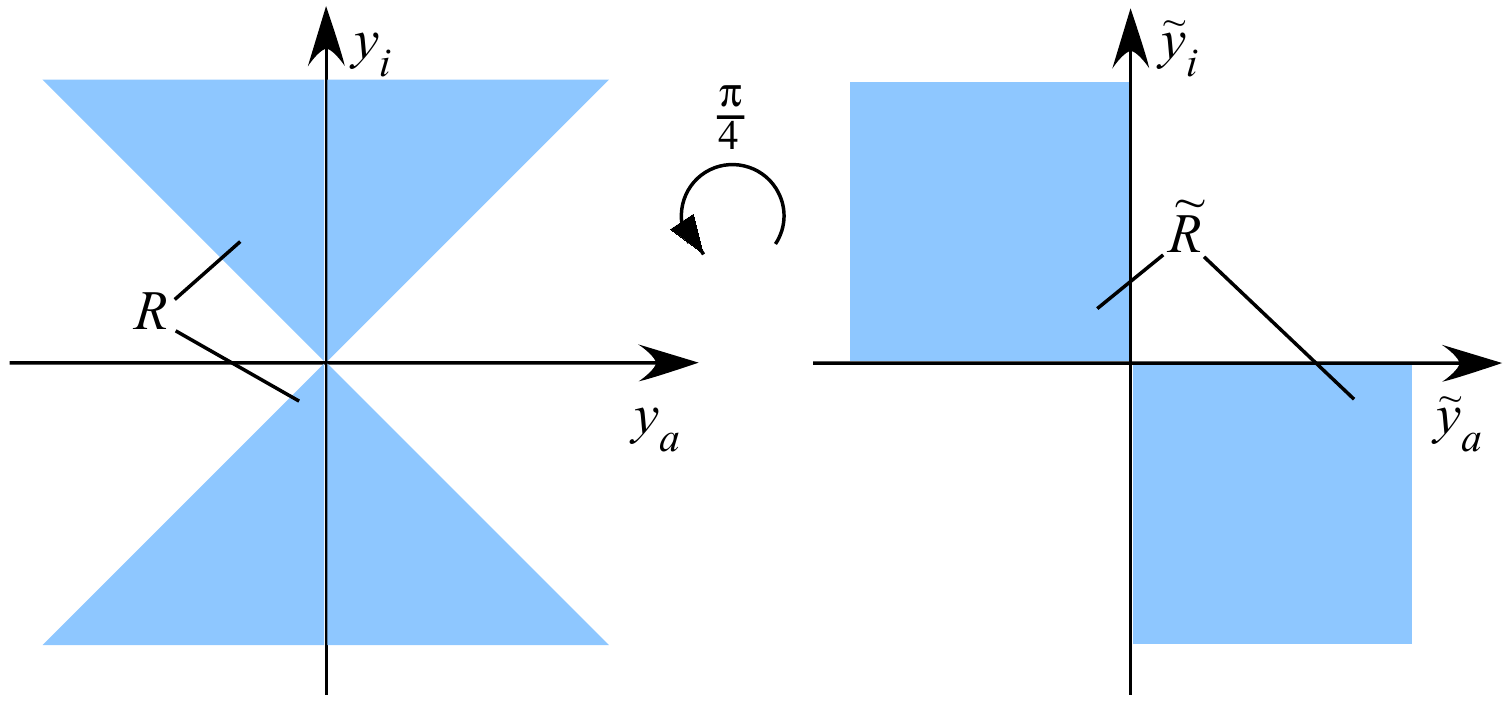}
\vspace*{-5mm}\caption{Illustration of the original (left-hand side) and rotated (right-hand side) regions of interest.}
\label{fig:prob_det}
\end{figure}

Due to the independence of $y_a$ and $y_i$ and the fact that they have equal variances, the distribution of $[y_a \, y_i]^T$ does not change under an orthogonal transformation $\boldsymbol{W}$, i.e., we can rotate the region $R$ without modifying the probability as
\begin{equation}
P \left( 
\begin{bmatrix}
y_a\\
y_i
\end{bmatrix} \in R
\right) = 
P \left( 
\boldsymbol{W}\begin{bmatrix}
y_a\\
y_i
\end{bmatrix} \in \widetilde{R}
\right),
\nonumber
\end{equation}
where $\widetilde{R} = \left\{ \boldsymbol{W} \boldsymbol{r}:  \, \boldsymbol{r} \in R \right\}$. So, by rotating the region $R$ by $\frac{\pi}{4}$ counterclockwise around the origin (see the right-hand side of Fig.~\ref{fig:prob_det}), i.e., by setting
\begin{equation}
\boldsymbol{W} =
\begin{bmatrix}
\cos\left(\frac{\pi}{4}\right) & -\sin\left(\frac{\pi}{4}\right)\\
\sin\left(\frac{\pi}{4}\right) & \cos\left(\frac{\pi}{4}\right)
\end{bmatrix} =
\frac{1}{\sqrt{2}}
\begin{bmatrix}
1 & -1\\
1 & 1
\end{bmatrix},
\end{equation}
the integral in~\eqref{eq:orig_int} can be calculated fairly straightforwardly. For the sake of notation simplicity, let us define
\begin{equation}
\boldsymbol{W}
\begin{bmatrix}
y_a\\
y_i
\end{bmatrix} =
\begin{bmatrix}
\tilde{y}_a\\
\tilde{y}_i
\end{bmatrix},
\nonumber
\end{equation}
where $\tilde{y}_a \sim \mathcal{N}(\tilde{\mu}_{y_a}, \sigma_{y_a}^2)$ and $\tilde{y}_i \sim \mathcal{N}(\tilde{\mu}_{y_i}, \sigma_{y_i}^2)$, with $\tilde{\mu}_{y_a} = \frac{1}{\sqrt{2}}\left( \mu_{y_a} - \mu_{y_i} \right)$ and $\tilde{\mu}_{y_i} = \frac{1}{\sqrt{2}}\left( \mu_{y_a} + \mu_{y_i} \right)$. Hence, the integral in~\eqref{eq:orig_int} is solved as
\begin{equation}
\begin{array}{l}
P\left( \, \abs{y_a} < \abs{y_i} \right) =\\
P\left( \tilde{y}_a < 0 \, \cap \, \tilde{y}_i > 0 \right) + P\left( \tilde{y}_a > 0 \, \cap \, \tilde{y}_i < 0 \right) =\\
P\left( \tilde{y}_a - \frac{\tilde{\mu}_{y_a}}{\sigma_{y_a}} < - \frac{\tilde{\mu}_{y_a}}{\sigma_{y_a}} \right) P\left( \tilde{y}_i - \frac{\tilde{\mu}_{y_i}}{\sigma_{y_i}} > - \frac{\tilde{\mu}_{y_i}}{\sigma_{y_i}} \right) +\\
P\left( \tilde{y}_a - \frac{\tilde{\mu}_{y_a}}{\sigma_{y_a}} > - \frac{\tilde{\mu}_{y_a}}{\sigma_{y_a}} \right) P\left( \tilde{y}_i - \frac{\tilde{\mu}_{y_i}}{\sigma_{y_i}} < - \frac{\tilde{\mu}_{y_i}}{\sigma_{y_i}} \right) =\\
Q\left( \frac{\tilde{\mu}_{y_a}}{\sigma_{y_a}} \right) Q\left( -\frac{\tilde{\mu}_{y_i}}{\sigma_{y_i}} \right) + Q\left( -\frac{\tilde{\mu}_{y_a}}{\sigma_{y_a}} \right) Q\left( \frac{\tilde{\mu}_{y_i}}{\sigma_{y_i}} \right).
\end{array}
\nonumber
\end{equation}

Finally, the lower bound on the probability of attacker detection in~\eqref{eq:prob_det_gen} is given by
\begin{equation}
\begin{array}{l}
P_D \geq LPD1 =\\
1 \hspace*{-0.5mm} - \hspace*{-0.5mm} \left( \hspace*{-0.5mm} \displaystyle\sum_{i: \, i\neq a} \hspace*{-1mm} Q\left( \frac{\tilde{\mu}_{y_a}}{\sigma_{y_a}} \right) \hspace*{-0.5mm} Q\left( -\frac{\tilde{\mu}_{y_i}}{\sigma_{y_i}} \right) \hspace*{-1mm} + \hspace*{-0.5mm} Q\left( -\frac{\tilde{\mu}_{y_a}}{\sigma_{y_a}} \right) \hspace*{-0.5mm} Q\left( \frac{\tilde{\mu}_{y_i}}{\sigma_{y_i}} \right) +\right.\\ \left. 1- \left( Q\left( \frac{\tau+\mu_{y_a}}{\sigma_{y_a}} \right) + Q\left( \frac{\tau-\mu_{y_a}}{\sigma_{y_a}} \right) \right) \hspace*{-1mm} \right).
\end{array}
\label{eq:LPD1}
\end{equation}

The lower bound in~\eqref{eq:LPD1} is a union bound, which might not be sufficiently tight in all scenarios. Therefore, we can make it tighter as
\begin{equation}
P_D \geq LP_D = \max\left\{ LPD1, LPD2 \right\},
\label{eq:LPD}
\end{equation}
where $LPD2$ is another lower bound on $P_D$ derived as
\begin{equation}
\begin{array}{l}
P_D \geq LPD2 =\\
P\left(e_a > \tau \right) \times P\left( \underset{i: \, i \neq a}{\max}\left\{ e_i \right\} \leq \tau \right) =\\
\left[ Q\left( \frac{\tau+\mu_{y_a}}{\sigma_{y_a}} \right) + Q\left( \frac{\tau-\mu_{y_a}}{\sigma_{y_a}} \right) \right] \times\\
\displaystyle\prod_{i: i\neq a} \left[ 1 - \left( Q\left( \frac{\tau+\mu_{y_i}}{\sigma_{y_i}} \right) + Q\left( \frac{\tau-\mu_{y_i}}{\sigma_{y_i}} \right) \right) \right].
\end{array}
\label{eq:LPD2}
\end{equation}

Similarly, we can upper-bound $P_D$ as
\begin{equation}
\begin{array}{l}
P_D \leq UP_D = P\left(e_a > \tau \right) =\\
Q\left( \frac{\tau+\mu_{y_a}}{\sigma_{y_a}} \right) + Q\left( \frac{\tau-\mu_{y_a}}{\sigma_{y_a}} \right).
\end{array}
\label{eq:UPD}
\end{equation}

\section{Generalized Trust Region Sub-problems}
\label{app:GTRS}

GTRS is characterized by minimizing a quadratic objective function over a quadratic constraint, which makes the problem non-convex in general. Even so, it is a monotonically decreasing function over an interval that we can calculate fairly easily, which is why GTRS is convenient for solving via bisection, see Fig.~\ref{fig:GTRS}.
\begin{figure}
\centering
\hspace*{-0mm}\includegraphics[width=\linewidth]{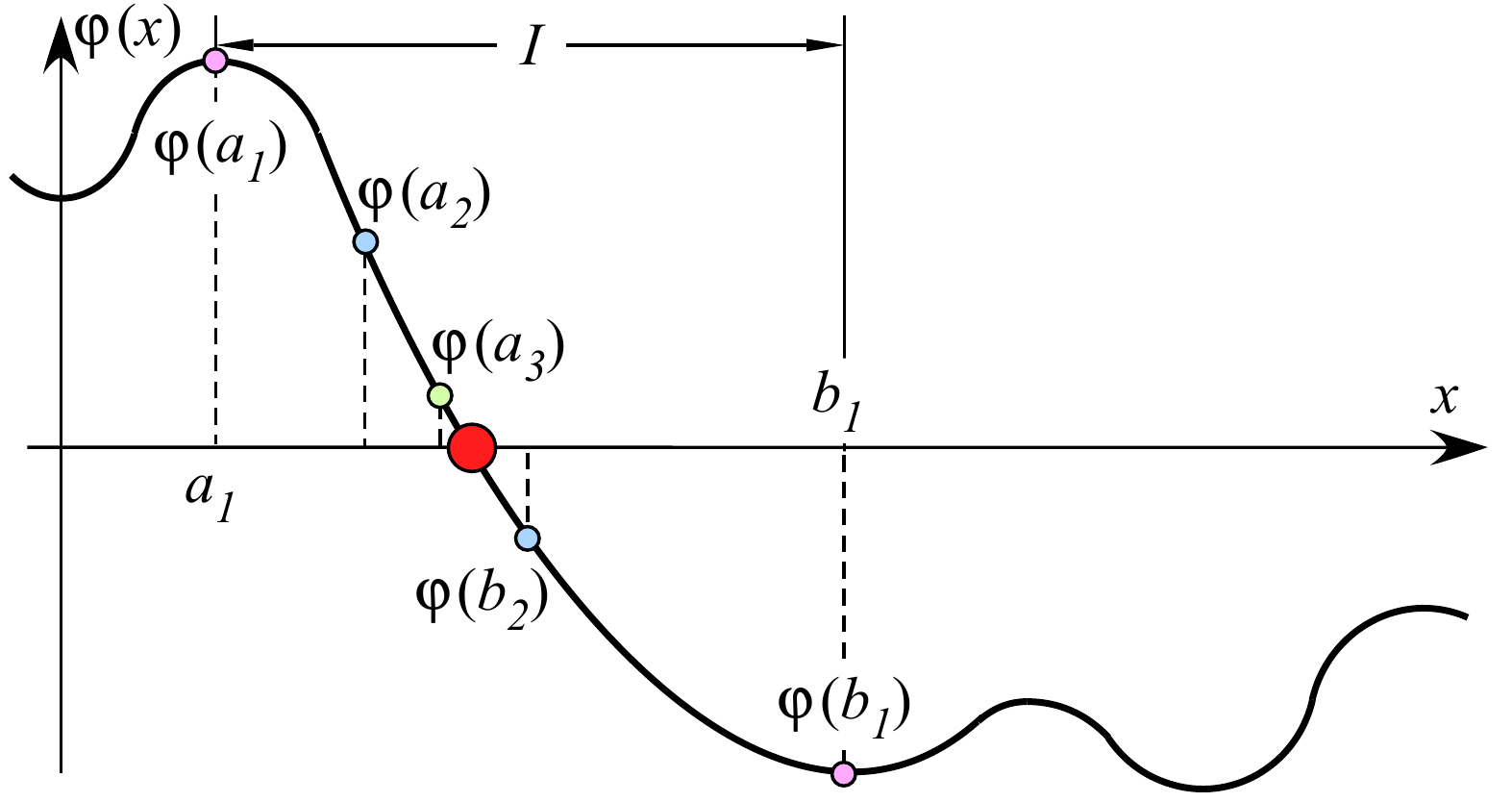}
\vspace*{-5mm}\caption{Illustration of the suitability of GTRS for solving via bisection.}
\label{fig:GTRS}
\end{figure}

According to \cite[Theorem 3.2]{More:1993}, $\boldsymbol{y} \in \mathbb{R}^{q+1}$ is an optimal solution to~\eqref{eq:sr-wls} if, and only if, there is $\lambda \in \mathbb{R}$ such that
\begin{equation}
\begin{array}{c}
\left(\boldsymbol{H}^T\boldsymbol{W}^T\boldsymbol{W}\boldsymbol{H} + \lambda\boldsymbol{F}\right)\boldsymbol{y} = \boldsymbol{H}^T\boldsymbol{W}^T\boldsymbol{W}\boldsymbol{h} - \lambda \boldsymbol{f}\\
\boldsymbol{y}^T \boldsymbol{F} \boldsymbol{y} + 2 \boldsymbol{f}^T \boldsymbol{y} = 0\\
\boldsymbol{H}^T\boldsymbol{W}^T\boldsymbol{W}\boldsymbol{H} + \lambda\boldsymbol{F} \succeq \boldsymbol{0}.
\end{array}
\nonumber
\end{equation}
Therefore, the optimal solution of~\eqref{eq:sr-wls} is
\begin{equation}
\hat{\boldsymbol{y}}(\lambda) = \left(\boldsymbol{H}^T\boldsymbol{W}^T\boldsymbol{W}\boldsymbol{H} + \lambda\boldsymbol{F}\right)^{-1} \left(\boldsymbol{H}^T\boldsymbol{W}^T\boldsymbol{W}\boldsymbol{h} - \lambda \boldsymbol{f}\right),
\nonumber
\end{equation}
where $\lambda$ is the unique solution of
\begin{equation}
\varphi(\lambda) = 0, \,\,\, \lambda \in I,
\nonumber
\end{equation}
the function $\varphi(\lambda) = \hat{\boldsymbol{y}}(\lambda)^T \boldsymbol{F} \hat{\boldsymbol{y}}(\lambda) + 2 \boldsymbol{f}^T \hat{\boldsymbol{y}}(\lambda)$ and the interval $I = \left( -\frac{1}{\lambda_{\text{max}}(\boldsymbol{F}, \boldsymbol{H}^T\boldsymbol{W}^T\boldsymbol{W}\boldsymbol{H})}, \infty \right)$, with $\lambda_{\text{max}}$ being the maximum eigenvalue of $\left(\boldsymbol{H}^T\boldsymbol{W}^T\boldsymbol{W}\boldsymbol{H}\right)^{-\frac{1}{2}} \boldsymbol{F} \left(\boldsymbol{H}^T\boldsymbol{W}^T\boldsymbol{W}\boldsymbol{H}\right)^{-\frac{1}{2}}$.

\section{Generalized Likelihood Ratio Test for WLS}
\label{app:GLRT}

For the purpose of testing, we assume two hypotheses, i.e., $H_0: d_{i,k} = \| \boldsymbol{x} - \boldsymbol{a}_i \| + n_{i,k}$ and $H_1: d_{i,k} = \| \boldsymbol{x} - \boldsymbol{a}_i \| + \delta_i + n_{i,k}$. Then, according to the two hypotheses, one can write the respective likelihood functions as follows.
\begin{equation}
p\left(\boldsymbol{d_{i}}|H_0\right) = c \exp\left\{ \frac{1}{2\sigma^2}\displaystyle\sum_{k=1}^K \left( d_{i,k} - \| \boldsymbol{x} - \boldsymbol{a}_i \| \right)^2 \right\},
\nonumber
\end{equation}
\begin{equation}
p\left(\boldsymbol{d_{i}}|H_1\right) = c \exp\left\{ \frac{1}{2\sigma^2}\displaystyle\sum_{k=1}^K \left( d_{i,k} - \| \boldsymbol{x} - \boldsymbol{a}_i \| - \delta_i \right)^2 \right\},
\nonumber
\end{equation}
with $c=\frac{1}{\left(2\pi \sigma\right)^{K/2}}$.

Therefore, according to GLRT \cite[Ch. 4]{Kay:1998}, we have that
\begin{equation}
\frac{p\left(\boldsymbol{d_{i}}|\widehat{\delta}_i, H_1\right)}{p\left(\boldsymbol{d_{i}}|H_0\right)}\displaystyle\mathop{\lessgtr}_{H_1}^{H_0} \gamma,
\label{eq:GLRT_orig}
\end{equation}
where $$\widehat{\delta}_i = \frac{\sum_{k=1}^K \left( d_{i,k} - \| \hat{\boldsymbol{x}}^{(WLS)} - \boldsymbol{a}_i \| \right)}{K}$$ is the ML estimate of $\delta_i$, $\hat{\boldsymbol{x}}^{(WLS)}$ is the target estimate obtained by solving the WLS in \cite{Mukhopadhyay:2018}, and $\gamma$ represents a threshold. After some simple algebraic manipulations, it can be shown that \eqref{eq:GLRT_orig} boils down to
\begin{equation}
\widehat{\delta}_i \displaystyle\mathop{\lessgtr}_{H_1}^{H_0} \sqrt{\frac{2\sigma^2\ln(\gamma)}{K}}.
\label{eq:GLRT}
\end{equation}

The probability of false alarm can then be written as
\begin{equation}
P_{FA} = P\left( r_i > \sqrt{\frac{2\sigma^2\ln(\gamma)}{K}} \, \Bigg| \, H_0 \right) = Q \left( \sqrt{2\ln(\gamma)} \right),
\label{eq:PFA}
\end{equation}
where $$r_i = \frac{1}{K} \sum_{k=1}^K \left( d_{i,k} - \| \boldsymbol{x} - \boldsymbol{a}_i \| \right),$$ i.e., $r_i \sim \mathcal{N}\left( 0, \frac{\sigma^2}{K} \right)$, under the hypothesis $H_0$. Hence, for a chosen value of $P_{FA}$ in \eqref{eq:PFA}, one can easily calculate the value of $\gamma$ in order to solve \eqref{eq:GLRT}. Finally, one can calculate the probability of detection according to GLRT, when the estimates of the unknown parameters are obtained through WLS, as
\begin{equation}
P_D = Q \left( \sqrt{2\ln(\gamma)} - \frac{\widehat{\delta}_i\sqrt{K}}{\sigma} \right).
\nonumber
\end{equation}

\end{appendices}





\vspace*{5mm}
\begin{wrapfigure}{l}{25mm} 
\includegraphics[width=1in,height=1.25in,clip,keepaspectratio]{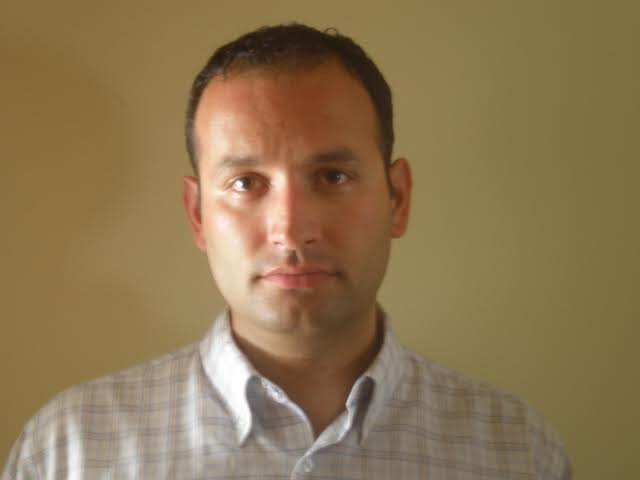}
\end{wrapfigure}\par
\textbf{Marko Beko} was born in Belgrade, Serbia, on November 11, 1977. He received the PhD degree in electrical and computer engineering from Instituto Superior T\'{e}cnico, Lisbon, Portugal, in 2008. He received the title of ``Professor com Agrega\c{c}\~{a}o'' in Electrical and Computer Engineering from Universidade Nova de Lisboa, Lisbon, Portugal, in 2018.  Currently, he is a Full Professor (Professor Catedr\'{a}tico) at the Universidade Lus\'{o}fona de Humanidades e Tecnologias, Lisbon, Portugal. He serves as an Associate Editor for the IEEE Open Journal of the Communications Society and Elsevier Journal on Physical Communication. He is the winner of the 2008 IBM Portugal Scientific Award.\par

\vspace*{5mm}
\begin{wrapfigure}{l}{25mm} 
\includegraphics[width=1in,height=1.25in,clip,keepaspectratio]{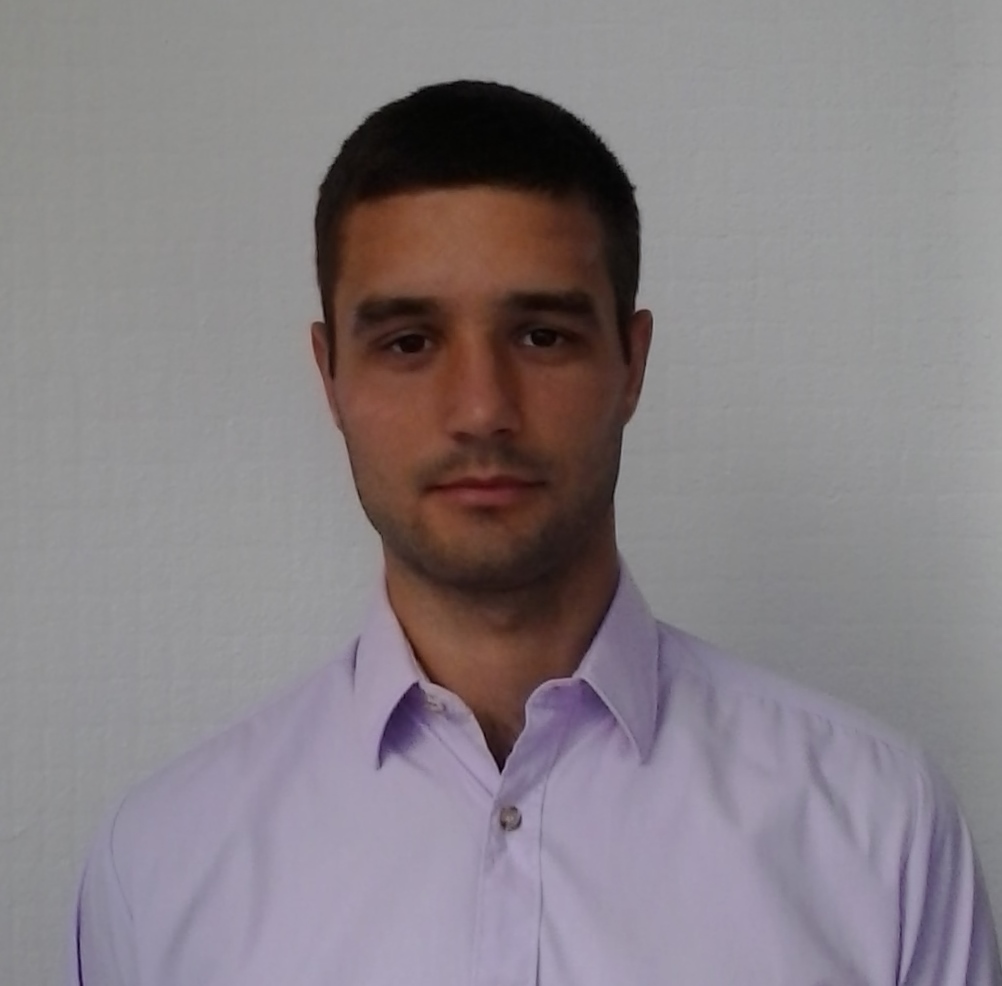}
\end{wrapfigure}\par
\textbf{Slavisa Tomic} received the M.S. degree in traffic engineering according to the postal traffic and telecommunications study program from University of Novi Sad, Serbia, in 2010, and the PhD degree in electrical and computer engineering from University Nova of Lisbon, Portugal, in 2017. He is currently an Assistant Professor at the Universidade Lus\'{o}fona de Humanidades e Tecnologias, Lisbon, Portugal. His research interests include target localization in wireless sensor networks, and non-linear and convex optimization.\par

\end{document}